\documentclass[a4paper,12pt]{scrartcl}
\usepackage[utf8]{inputenc}
\usepackage{amssymb,amsmath,amsthm,mathrsfs,enumitem}
\usepackage[top=2.5cm,left=2.5cm,right=2.5cm,bottom=3cm,footskip=1.5cm]{geometry}
\usepackage{hyperref}
\hypersetup{hidelinks}
\usepackage{pgfplots}
\usepackage[affil-it]{authblk}
\usepgfplotslibrary{fillbetween}
\usepgfplotslibrary{polar}
\usetikzlibrary{arrows.meta}
\usepackage{caption}
\usepackage{subcaption}
\newcommand{\real}{\mathbb{R}}

\newcommand{\Mwp}{\mathcal{M}}
\newcommand{\Uwp}{\mathcal{U}}
\newcommand{\Mwop}{M}

\newcommand{\pb}{b}

\newcommand{\pt}{p_c}
\newcommand{\Ip}{I^+_\Mwp}
\newcommand{\Ipm}{I^\pm_\Mwp}

\newcommand{\Imi}{I^-_\Mwp}
\newcommand{\Jp}{J^+_\Mwp}
\newcommand{\Jpm}{J^\pm_\Mwp}
\newcommand{\Jm}{J^-_\Mwp}
\newcommand{\IpM}{I^+_\Mwop}
\newcommand{\IpmM}{I^\pm_\Mwop}
\newcommand{\ImpM}{I^\mp_\Mwop}
\newcommand{\ImM}{I^-_\Mwop}
\newcommand{\IpF}{I^+_M}
\newcommand{\ImF}{I^-_M}
\newcommand{\IpXZ}{I^+_{XZ}}
\newcommand{\ImXZ}{I^-_{XZ}}
\newcommand{\JpM}{J^+_\Mwop}
\newcommand{\JpmM}{J^\pm_\Mwop}
\newcommand{\JmM}{J^-_\Mwop}
\newcommand{\giso}{g_\mathrm{iso}}

\newcommand{\hpert}{h^{(1)}}

\newcommand{\hess}{\operatorname{hess}}
\newcommand{\sign}{\operatorname{sign}}
\newcommand{\rrho}{\alpha}
\newcommand{\mina}{a}
\newcommand{\maxa}{A}

\theoremstyle{plain}
\newtheorem{thm}{Theorem}
\newtheorem{conj}{Conjecture}
\newtheorem{lem}{Lemma}
\newtheorem{prop}{Proposition}

\newtheorem*{simdiag}{Simultaneous Diagonalization Theorem {\normalfont{\cite[Thm. 13.4.3]{GaHo}}}}
\newtheorem*{Morselem}{Morse Lemma {\normalfont{\cite[Lem. 2.2]{Mil}}}}
\newtheorem{cond}{Condition}

\theoremstyle{definition}
\newtheorem{defn}{Definition}

\newtheorem{ex}{Example}

\theoremstyle{remark}
\newtheorem*{rem}{Remark}

\title{Topology change with Morse functions}
\subtitle{progress on the Borde--Sorkin conjecture}
\author{Leonardo Garc\'ia-Heveling
 \thanks{Mathematics Department, Radboud University, Nijmegen, The Netherlands \\ \normalfont \ttfamily{l.heveling@math.ru.nl}}}

\begin{document}

\date{\vspace{-5ex}}
\maketitle

\begin{abstract}
 Topology change is considered to be a necessary feature of quantum gravity by some authors, and impossible by others. One of the main arguments against it is that spacetimes with changing spatial topology have bad causal properties. Borde and Sorkin proposed a way to avoid this dilemma by considering topology changing spacetimes constructed from Morse functions, where the metric is allowed to vanish at isolated points. They conjectured that these Morse spacetimes are causally continuous (hence quite well behaved), as long as the index of the Morse points is different from $1$ and $n-1$. In this paper, we prove a special case of this conjecture. We also argue, heuristically, that the original conjecture is actually false, and formulate a refined version of it.
 \\
 
\end{abstract}

\section{Introduction}

In General Relativity, by solving the initial value problem for Einstein's Equations, one finds the time evolution of the spacetime metric. In this picture, the topology of the constant time slices always remains the same. More precisely, the maximal globally hyperbolic developement of some initial data $V$ is, on the level of topology (Geroch \cite{Ger2}) and differentiable structure (Bernal and S\'anchez \cite{BeSa1,BeSa2}), simply $V \times \real$. The question remains whether this rigid product structure is desirable, or whether we should allow the topology to change over time as well. There are several instances where topology change is desirable. The dynamical creation of a wormhole, for example, is necessarily a topology changing process, as it involves attaching a handle to space. Already in 1957, Wheeler argued that quantum fluctuations of spacetime should modify the topology \cite{Whe}. Moreover, in certain approaches to Quantum Gravity, instead of considering the deterministic evolution of a spatial slice under the Einstein Equations, the idea is to find the transition probability between two spatial slices $V_1,V_2$. This is done by computing a path-integral over all cobordisms betweeen $V_1$ and $V_2$; that is, manifolds $\Mwp$ with boundary $\partial \Mwp = V_1 \sqcup V_2$. These cobordisms also have to be equipped with a Lorentzian (or, in Euclidean Quantum Gravity, Riemannian) metric, and possibly satisfy some additional conditions. It is then natural to think that the transition probability between $V_1$ and $V_2$ might be non-zero also when $V_1$ and $V_2$ are not homeomorphic, as long as appropriate cobordisms exist. We refer to \cite{Dow,Sor} for further discussion on the role of topology change within Quantum Gravity.

Which properties should a Lorentzian cobordism satisfy, in order to consider it physically reasonable? In this paper, we will focus on the case of compact cobordisms (i.e.\ spatially closed universes). Geroch \cite{Ger} showed that any non-trivial (meaning with $V_1 \neq V_2$) compact Lorentzian cobordism must contain closed timelike curves. Because of this, the only way to have topology change without time travel is by allowing the spacetime metric to degenerate at certain singular points \cite{Kun,Yod1}. Notice that the case of non-compact time slices is less restricted, with examples of topology change without closed timelike curves and without singular points obtained by multiple authors (see S\'{a}nchez \cite{San2} for the most recent ones and for the overview of previous work on p.~16).

One interpretation is to consider the singular points as naked singularities, and not as points in the spacetime manifold. In this paper, however, we do the opposite: we consider the singular points as points in the spacetime, where nothing special happens, except that, in some sense, the topology change happens there. Our point of view implies that the spacetime metric is not Lorentzian everywhere, but this is not so bad, since the metric is not a physical observable in itself. Indeed, we will show that the causal and length structures can be satisfactorily generalized to include the singular points (some work on the curvature has also been done \cite{LoSo}). Still, allowing degenerate metrics does introduce many new questions and problems (irrespective of our point of view on the singular points). Already in the 1980s, Anderson and DeWitt showed that on their famous ``trousers spacetime", quantum fields create infinite bursts of energy in the presence of singular points \cite{AnDW}. This result was later refined and confirmed in Manogue et al.\ \cite{MCT} and Buck et al.\ \cite{BDJS}. The aim of subsequent work was to impose additional conditions that avoid such pathologies.

A concrete and very useful construction of degenerate Lorentzian metrics on cobordisms was given by Yodzis \cite{Yod1,Yod2} using Morse functions. This idea was further developed by Sorkin \cite{Sor2} and collaborators \cite{LoSo,DGS1,DGS2,BDGSS,DoGa,DoSu}, under the name of \emph{Morse geometries}. We continue this approach in the present paper.

The construction of a Morse geometry is as follows. For simplicity, all objects are assumed to be smooth. Let $\Mwp$ be a compact cobordism of dimension $n$, $h$ a Riemannian metric, $\zeta>1$ a constant, and $f$ a Morse function. Recall that a smooth function $f \colon \Mwp \to \real$ is called a Morse function if all its critical points (where $df=0$) are non-degenerate (not to be confused with the (non-)degeneracy of the spacetime metric). This is equivalent to saying that around each critical point, there exist coordinates $x^i$ such that
\begin{equation} \label{Morsefct}
    f = \frac{1}{2} \sum_i a_i (x^i)^2,
\end{equation}
where $a_i \neq 0$ are constants. It follows, in particular, that the critical points are isolated. The \emph{index} of a critical point is defined as the number of negative $a_i$ (see \cite{Mil} for more details). Louko and Sorkin define the \emph{Morse metric} corresponding to $(\Mwp,h,f,\zeta)$ by
\begin{equation} \label{metricinv}
 g = \left\Vert df \right\Vert_h^2 h - \zeta df \otimes df,
\end{equation}
which, in coordinates, gives
\begin{equation} \label{metriccoo}
 g_{\mu\nu} := (h^{\alpha \beta} \partial_\alpha f \partial_\beta f) h_{\mu\nu} - \zeta \partial_\mu f \partial_\nu f.
\end{equation}
Let $\Mwop = \Mwp \setminus \left( \partial \Mwp \cup \{p_i\}_i \right)$, where $\{p_i\}_i$ is the set of critical points of $f$. By abuse of notation, we call the restriction of $g$ to $\Mwop$ also $g$. Since $df$ vanishes only at the critical points, $g$ is Lorentzian on $\Mwop$, and the pair $(\Mwop,g)$ forms a spacetime in the usual sense. It is clear from \eqref{metricinv} that $f$ is a time function on $(\Mwop,g)$, when choosing the time orientation to be given by the gradient vector field $\nabla^\alpha f := h^{\alpha \beta} \partial_\beta f$. Following the nomenclature of Borde et al.\ \cite{BDGSS}, we call $(\Mwop,g)$ a \emph{Morse spacetime} and $(\Mwp,h,f,\zeta)$ a Morse geometry\footnote{In \cite{DGS1}, the inverted nomenclature is used.}. It is known that for any pair of connected $3$-manifolds, there exists a Morse geometry interpolating between them \cite{DoSu,DGS1}.

According to the following two conjectures, the infinite bursts of energy found by Anderson and DeWitt are only present on certain topology-changing spacetimes, but not on others.

\begin{conj}[Sorkin] \label{Sconj}
 A quantum field propagating on a Morse geometry $(\Mwp,h,f,\zeta)$ has an unphysical singular behaviour if and only if the Morse spacetime $(M,g)$ is causally discontinuous.
\end{conj}

\begin{conj}[Borde--Sorkin] \label{BSconj}
 The Morse spacetime $(\Mwop, g)$ induced by a Morse geometry $(\Mwp,h,f,\zeta)$ is causally continuous if and only if all critical points of $f$ have index different from $1$ and $n-1$.
\end{conj}

Recall that causal continuity roughly means that the past and future $I^\pm(p)$ varies continuously with the point $p$ (see Appendix \ref{appcc} for details). Causal continuity was introduced in Hawking and Sachs \cite{HaSa} as a minimal requirement for a spacetime to be physically reasonable, for reasons unrelated to quantum theory. Thus Conjecture \ref{BSconj} is also interesting beyond the obvious link to Conjecture \ref{Sconj}. Conjecture \ref{Sconj} is mentioned as early as 1990 in Sorkin \cite{Sor2}, while the earliest reference for Conjecture \ref{BSconj} is an indirect source (Dowker and Garcia \cite{DoGa} from 1998). Both of them remain open to this day. Conjecture \ref{BSconj} has seen important progress trough the works of Borde, Dowker, Garcia, Sorkin and Surya \cite{BDGSS,DGS2,DGS1}. In this paper, we contribute to this effort by showing the following special case:

\begin{thm} \label{mainthm}
 Let $(\Mwp,h,f,\zeta)$ be a Morse geometry of dimension $n$ with a single critical point $\pt \in \Mwp$. Suppose that $\pt$ has index $\lambda \neq 0,1,n-1,n$, and is contained in a coordinate neighborhood where
 \begin{align} \label{fhintro}
    &f = \frac{1}{2} \sum_i a_i (x^i)^2,
    &h = \sum_i (dx^i)^2,
\end{align}
for some real constants $a_i \neq 0$ satisfying
\begin{equation} \label{largezetaintro}
\frac{1}{\zeta} < \left\vert \frac{a_i}{a_j} \right\vert < \zeta \quad \textrm{and} \quad \frac{5}{8}  \leq \left\vert \frac{a_i}{a_j} \right\vert \leq \frac{8}{5} \quad \textrm{for all $i,j$.}
\end{equation}
Then the corresponding Morse spacetime $(\Mwop,g)$ is causally continuous.
\end{thm}

In Section \ref{secfull} (Proposition \ref{proploc}) we will show that one can \emph{always} find coordinates where \eqref{fhintro} holds, up to adding a perturbation to $h$ which vanishes at $\pt$. Moreover, we conjecture that the first part of \eqref{largezetaintro} is sharp, in the sense that its violation leads to causal discontinuity (see Example \ref{ex2} and Conjecture \ref{BSconjmod}).

Combining Theorem \ref{mainthm} with previous results by other authors (fleshed out below), we can summarize the current progress on Conjecture \ref{BSconj} in the next theorem.

\begin{thm} \label{thm2}
 Let $(\Mwp,h,f,\zeta)$ be a Morse geometry of dimension $n \geq 2$, and $(\Mwop,g)$ the corresponding Morse spacetime. Assume $f$ has a single critical point for each critical value.
 \begin{enumerate}
     \item If $f$ has at least one critical point of index $\lambda = 1, n-1$, then $(\Mwop,g)$ is causally discontinuous.
     \item If each critical point of $f$ has index $\lambda=0,n$, or has any index $\lambda \neq 1, n-1$ and is contained in a neighborhood as in Theorem \ref{mainthm}, then $(\Mwop,g)$ is causally continuous.
 \end{enumerate}
\end{thm}

The case $\lambda = 0,n$ in part (ii) was solved in Borde et al.\ \cite{BDGSS}, along with the special case of Theorem \ref{mainthm} corresponding to $| a_i | = 1$ for all $i$. Part (i) of Theorem \ref{thm2} was shown in Dowker et al.\ \cite{DGS1}. Also in \cite{DGS1}, it was shown that the case of multiple critical points (as long as there is only one per critical value) reduces to the case of a single critical point: the Morse spacetime is causally continuous if and only if every critical point has a causally continuous neighborhood.

The proof of Theorem \ref{mainthm} is contained in Section \ref{secproof}. In Section \ref{secfull}, we discuss the necessity of our assumptions, and possible generalizations of our proof. Based on this discussion, we propose a modified version of the Borde--Sorkin conjecture in Section \ref{secconc}, where we also give concluding remarks. Appendix \ref{appnbhd} contains results of \cite{BDGSS} that we need in our proofs, and Appendix \ref{appcc} gives some background on causal continuity.

\section{Proof of Theorem \ref{mainthm}} \label{secproof}

Before starting, let us give a brief outline of the proof. Recall from the introduction that the case of $| a_i | = 1$ for all $i$ has already been solved in Borde et al.\ \cite{BDGSS}, a result that we build upon. While in the case of $| a_i | = 1$ there are a lot of symmetries, which allow for good coordinate choices (see Appendix \ref{appnbhd}), this is no longer true in the general case. Our strategy is to extend the causal structure from $(\Mwop,g)$ to $\Mwp$, in a way that preserves its most important properties: openness of the chronological relation $I^+$, the push-up principle $J^+(I^+(q)) = I^+(J^+(q)) = I^+(q)$, and the properties of limits of causal curves. Once these properties are proven, causal continuity follows almost immediately, as it would in Minkowski spacetime.

The most difficult to establish, out of the three properties, is the openness of the chronological relation. We do this in Subsection \ref{secopen}. The argument is based on reduction to the $| a_i | = 1$ case. Once openness of the chronological relation is established, the rest of the proof can be performed without the need to make any coordinate choices whatsoever, and without further use of the assumptions \eqref{fhintro} and \eqref{largezetaintro}. This second part of the proof is contained in Subsection \ref{secpushup}. It requires heavy use of the limit curve theorems of Minguzzi \cite{Min}.

\subsection{Notation and first steps}

Throughout this section, we assume that $(\Mwp,h,f,\zeta)$ is a Morse geometry of dimension $n \geq 4$, with a single critical point $\pt$ of index $\lambda \neq 0,1,n-1,n$ lying in the interior of $\Mwp$. As in the introduction, we write $\Mwop := \Mwp \setminus \left( \partial \Mwp \cup \{\pt\} \right)$, and $g$ denotes the metric \eqref{metricinv}, which is Lorentzian on $\Mwop$ and degenerate-Lorentzian on $\Mwp$. We do not use Einstein's summation convention: all sums are written out, but without making explicit the summation limits. Hence $\sum_i$ means $\sum_{i=1}^n$, and similarly $\max_i$ means $\max_{i=1,...,n}$. For convenience, we refer to the hypothesis of Theorem \ref{mainthm} as Condition \ref{condmain}.

\begin{cond} \label{condmain}
 There exists an open set $\Uwp \subseteq \Mwp$ with $\pt \in \Uwp$, an open ball $\mathcal{B} \in \real^n$ around the origin, and a coordinate chart $\phi : \Uwp \to \mathcal{B}$ of $\Mwp$ such that $\phi(\pt) = 0$ and
 \begin{align} \label{cooscondmain}
    &f \circ \phi^{-1} = \frac{1}{2} \sum_{i} a_i (x^i)^2,
    &h \circ \phi^{-1} = \sum_{i} (dx^i)^2,
\end{align}
for some real constants $a_i \neq 0$. Moreover, setting
\begin{equation*}
 \zeta_c := \max_{i,j} \left\vert \frac{a_i}{a_j} \right\vert,
\end{equation*}
we have
\begin{equation} \label{largezeta}
    \zeta_c \leq \frac{8}{5} \qquad \textrm{and} \qquad \zeta_c < \zeta.
\end{equation}
\end{cond}

The value of $\zeta_c$ does not depend on the choice of coordinates, as long they satisfy \eqref{cooscondmain} (we give a detailed argument for this in Section \ref{secgennbhd}). We will usually suppress the coordinate map $\phi$ from the notation, and whenever we write $x^i$, it will refer to the coordinates as given by Condition \ref{condmain}. In these coordinates, the metric \eqref{metricinv} takes the form
\begin{equation} \label{metricx}
    g = \sum_{i,j} \left(a_i x^i dx^j \right)^2  - \zeta \left( \sum_k a_k x^k dx^k \right)^2.
\end{equation}

An important tool in our proof will be to reduce some computations to the case of \emph{isotropic neighborhoods} as studied in Borde et al.\ \cite{BDGSS} (see also Appendix \ref{appnbhd}). These are metrics where Condition \ref{condmain} is satisfied, but with the stronger requirement that $\vert a_i \vert = 1$ for all $i$. The following lemma gives us such an isotropic neighborhood metric $\giso$ with lightcones narrower than those of $g$.

\begin{lem} \label{lemgiso}
 Suppose that Condition \ref{condmain} is satisfied, and consider on $\Uwp$ the linear change of coordinates
 \begin{equation*}
     x^i \mapsto y^i := \sqrt{\vert a_i \vert} \zeta_c^{\frac{1}{4}} x^i.
 \end{equation*}
 Then the tensor given in these new coordinates by
 \begin{equation} \label{giso}
  \giso := \sum_{i,j} \left( y^i dy^j \right)^2  - \frac{\zeta}{\zeta_c} \left( \sum_k \sign(a_k) y^k dy^k \right)^2
 \end{equation}
 is a Lorentzian metric on $\Uwp \setminus \{\pt\}$, with lightcones narrower that those of $g$.
\end{lem}

\begin{proof}
 By \eqref{largezeta}, ${\zeta}/{\zeta_c} > 1$, and hence $\giso$ is a \emph{neighborhood metric} in the sense of Borde et al.\ (see Appendix \ref{appnbhd}). In particular, $\giso$ is Lorentzian everywhere except at the origin. In the $y$-coordinates, the metric \eqref{metricx} takes the form
 \begin{equation*}
     g = \frac{1}{\zeta_c} \sum_{i,j} \left\vert \frac{a_i}{a_j} \right\vert \left( y^i dy^j \right)^2  - \frac{\zeta}{\zeta_c} \left( \sum_k \sign(a_k) y^k dy^k \right)^2.
 \end{equation*}
 For a vector $V \in T\Uwp$ (with components $V^i$ in the $y$-coordinates), this means
 \begin{align*}
     g(V,V) &= \frac{1}{\zeta_c} \sum_{i,j} \left\vert \frac{a_i}{a_j} \right\vert \left( y^i V^j \right)^2  - \frac{\zeta}{\zeta_c} \left( \sum_k \sign(a_k) y^k V^k \right)^2 \\ &\leq \sum_{i,j} \left( y^i V^j \right)^2  - \frac{\zeta}{\zeta_c} \left( \sum_k \sign(a_k) y^k V^k \right)^2 \\ &= \giso(V,V)
 \end{align*}
 Hence if $\giso(V,V) \leq 0$, then also $g(V,V) \leq 0$. In other words, $\giso$ has narrower lightcones than $g$.
\end{proof}

Another crucial element in the proof will be the extension of the causal relation from $\Mwop$ to $\Mwp$. Let $\gamma \colon I \to \Mwp$ be a locally Lipschitz curve. By continuity, $\gamma^{-1}(\Mwop)$ is open in $\real$ and hence can be written as a union of intervals $\bigcup_i I_i$. If $\gamma \colon I_i \to M$ is future-directed (f.d.)\ causal for every $i$, then we say that $\gamma \colon I \to \Mwp$ is future-directed causal, and analogously for timelike and/or past-directed curves. This gives rise to a notion of futures and pasts $\Ipm(p), \Jpm(p)$ in $\Mwp$. Additionally, for $p \neq \pt$, we denote by $\IpmM(p), \JpmM(p)$ the usual past and future sets in the spacetime $(M,g)$.

Given a f.d. causal curve $\gamma \colon I \to \Mwp$ as above, since $f$ is a time function on $(M,g)$, $f \circ \gamma$ is strictly increasing on $\gamma^{-1}(M)$. Therefore $f$ is increasing along all of $\gamma$ and can only cross $\pt$ once, i.e.\ $\gamma^{-1}(\pt)$ is empty, a point, or a closed connected interval in $\real$. The following lemma tells us that furthermore no causal curve can be imprisoned in a neighborhood of $\pt$  (see \cite[pp.\ 61-62]{BEE} for the definition of non-imprisonment on non-degenerate spacetimes).

\begin{lem} \label{lempbdry}
 Let $\gamma \colon (a,b) \to M$ be a causal curve which is future inextendible in $M$. Then either $\lim_{s \to b} \gamma(s) = \pt$ or $\gamma$ runs into $\partial \Mwp$.
\end{lem}

\begin{proof}
 Because $f$ is a time function on $(M,g)$, $(M,g)$ is strongly causal. Then, by \cite[Prop.\ 3.13]{BEE}, given any compact set $K \subseteq M$, there exists $\delta>0$ such that $\gamma(s) \not\in K$ for all $s \in (b-\delta, b)$ (in other words, $\gamma$ must leave $K$ and never enter it again). Let $U \subseteq \Mwp$ be any open set (not necessarily connected) containing $\pt$ and $\partial \Mwp$. Then we can choose $K = \Mwp \setminus U$, and hence there exists $\delta>0$ such that $\gamma(s) \in U$ for all $s \in (b-\delta,b)$. Since $U$ was arbitrary, we are done. %picture possible
\end{proof}

\subsection{Openness of chronological pasts and futures} \label{secopen}

In this subsection, we characterize the past $\Imi(\pt)$ of the critical point $\pt$. Every statement has a time-reversed analogue for the future $\Ip(\pt)$ (which we do not write out explicitly). The following condition is very important. It states that if from a point $q \in M$ we can reach $\pt$ via timelike curves, then we can also reach a whole neighborhood of $\pt$. This is a well-known fact for spacetimes without singular points.

\begin{cond}[Openness of $\Ip$] \label{condopen}
 For every $q \in \Ipm(\pt)$ there exists a neighborhood $U$ of $\pt$ such that $U \setminus \{\pt\} \subseteq \ImpM(q)$.
\end{cond}

An important consequence of Condition \ref{condopen} is that the chronological relation is not altered by removing $\pt$.

\begin{lem}[{$\IpM = \Ip \cap \Mwop$}] \label{avoidp}
 Suppose Condition \ref{condopen} is satisfied. Then, for every $p \in \Mwop$ it holds that $\IpM(p) = \Ip(p) \cap \Mwop$.
\end{lem}

\begin{proof}
 The inclusion $\IpM(p) \subset \Ip(p) \cap \Mwop$ is trivial. It remains to show th eother direction. Let $q \in \IpM(p)$, and let $\gamma: [a,b] \to \Mwp$ be a timelike curve from $p$ to $q$. If $\gamma$ avoids $\pt$, there is nothing to prove. Hence suppose that $\gamma(c)=\pt$ for some $c$. Then $\pt \in \Ip(p) \cap \Imi(q)$, so by Condition \ref{condopen} we can find neighborhoods $U$, $V$ of $\pt$ such that $U \setminus \{ \pt \} \subseteq \IpM(p)$ and $V \setminus \{ \pt \} \subseteq \ImM(q)$. But then we can find a point $z \in U \cap V \setminus \{ \pt \}$, and a timelike curve $\sigma : [a,b] \to \Mwop$ from $p$ to $q$ passing though $z$.
\end{proof}

In Appendix \ref{appnbhd} (Lemma \ref{lemapp}), we show that Condition \ref{condopen} holds for the isotropic metric $\giso$, which is simpler than $g$, and has narrower lightcones (Lemma \ref{lemgiso}). Making use of this fact, we show through the following lemma that Condition \ref{condopen} also holds for our metric of interest $g$.

\begin{lem} \label{lemnbhd}
 Condition \ref{condmain} implies Condition \ref{condopen}.
\end{lem}

The rest of this subsection is dedicated to proving Lemma \ref{lemnbhd}. We start by discussing coordinate choices. Assume w.l.o.g.\ that we have ordered our coordinates $x^i$, where $f$,$h$ take the form \eqref{cooscondmain}, so that $a_i <0$ for $i=1,...,\lambda$ and $a_j >0$ for $j=\lambda+1,...,n$. We then define the following ``radial" coordinates
 \begin{align*}
  &r^2 := \frac{1}{2} \sum_{i=1}^{\lambda} -a_i (x^i)^2,
  &\rho^2 := \frac{1}{2} \sum_{j=\lambda+1}^{n} a_j (x^j)^2.
 \end{align*}
 By following the flow of the gradient vector $\nabla r$ (by which we mean the gradient taken with respect to $h$, so that $h(\nabla r, \cdot) = dr( \cdot)$) we get a diffeomorphism from $\real^{n-\lambda} \setminus \{0\}$ to $\real \times S^{\lambda - 1}$. This gives us a coordinate system $(r,\theta_1,...,\theta_{\lambda - 1})$ on $\real^{\lambda}$, where we view $\real^{\lambda}$ as the subspace spanned by the $x^i$ coordinates with $i=1,...,\lambda$. Essentially, all we have done is changing to polar coordinates, but it is important that we have done so in a way that the angular directions are $g$-orthogonal to the $r$-direction. We can do the same construction with $\rho$, obtaining coordinates $(\rho,\phi_1,...,\phi_{n-\lambda-1})$ on $\real^{n-\lambda}$. Furthermore, we have
 \begin{align}
  &f = \rho^2 - r^2,
  &\rrho := (\rho r)^\frac{1}{p}, \label{frrho}
 \end{align}
 where $p>0$ is a constant, $f$ is just our Morse function, and $\rrho$ is chosen so that $h(\nabla f, \nabla \rrho) = 0$. Using $(f,\rrho,\theta_1,...,\theta_\lambda,\phi_1,...,\phi_{n-\lambda-1})$ as our coordinates, the Euclidean metric $h$ takes the form
 \begin{align*}
     h = \frac{df^2}{\Vert \nabla f \Vert^2} + \frac{d\rrho^2}{\Vert \nabla \rrho \Vert^2} + h_\Theta + h_\Phi,
 \end{align*}
 and thus the Morse metric $g$ takes the form
 \begin{equation} \label{metricrrho}
     g = -(\zeta-1) df^2 + \frac{\Vert \nabla f \Vert^2}{\Vert \nabla \rrho \Vert^2} d\rrho^2 + \Vert \nabla f \Vert^2 (h_\Theta + h_\Phi).
 \end{equation}
 Here we have used that, by definition, $\Vert \nabla f \Vert = \Vert df \Vert$. Having chosen our coordinates, we now state a lemma that constitutes the most important step in the proof of Lemma \ref{lemnbhd}.

\begin{lem} \label{lemtip}
 Suppose Condition \ref{condmain} is satisfied. Let $q \in \Imi(\pt)$, and let $\gamma \colon [0,1] \to \Mwp$ be any f.d.\ timelike curve from $\gamma(0) = q$ to $\gamma(1) = \pt$, which we express in coordinates as
 \begin{equation} \label{gammaincoos}
     \gamma(s) = (f(s),\rrho(s), \Theta(s), \Phi(s)).
 \end{equation}
 Then, for every $0 < \varepsilon < \rrho(0)$, the curve $\sigma \colon [0,s_\varepsilon] \to \Mwp$ given by
 \begin{equation} \label{defsigma}
     \sigma(s) := (f(s),\rrho(s) - \varepsilon, \Theta(s), \Phi(s))
 \end{equation}
 is f.d.\ timelike. Here $s_\varepsilon := \min \{s \in (0,1) \mid \rrho(s) = \varepsilon \}$.
\end{lem}

\begin{proof}
 The statement is trivially true if $\rrho(0) = 0$ (since then there exist no suitable $\varepsilon$), and otherwise $s_\varepsilon$ is well-defined (the minimum exists) by continuity of $\rrho(s)$. Moreover, our choice of $\epsilon$ and $s_\varepsilon$ ensures that $\rrho(s)-\varepsilon \geq 0$ for all $s \in [0,s_\varepsilon]$, so that the curve $\sigma$ is also well-defined.
 
 Note that shifting $\rrho$ by a constant $\varepsilon$ while leaving $f,\Theta,\Phi$ fixed (as is done in \eqref{defsigma}), is equivalent to shifting both $\rho^2$ and $r^2$ by a quantity $\epsilon(s)$. We are going to show that $g(\dot\sigma(s),\dot\sigma(s)) \leq g(\dot\gamma(s),\dot\gamma(s)) < 0$. This will be done in multiple steps, corresponding to various terms in \eqref{metricrrho}.
 
 \textbf{Step 1} (Angular part)\textbf{.} Let $\pi_\Theta$ denote the orthogonal projection onto the subspace of the tangent space spanned by the $\Theta$ angular directions. We will show that
 \begin{equation*}
     h(\pi_\Theta \dot{\sigma},\pi_\Theta  \dot{\sigma}) \leq h(\pi_\Theta  \dot{\gamma},\pi_\Theta  \dot{\gamma}).
 \end{equation*}
 An analogous statement holds for $\pi_\Phi$. We proceed by computing $\pi_\Theta \dot{\sigma}$. Notice that shifting $r(s)^2$ to $r(s)^2 - \epsilon(s)$ means following the flow $F \colon \Mwp \times \real \to \Mwp$ of the vector field $\nabla r$ for a certain time $t(s) > 0$. Then
 \begin{equation} \label{flow}
     \dot\sigma(s) = DF(\gamma(s),t(s)) \dot\gamma(s) + \frac{\partial F}{\partial t} (\gamma(s),t(s)) \dot{t}.
 \end{equation}
 Similary, shifting $\rho^2$ means following the flow of $- \nabla \rho$. Notice also that
 \begin{equation*}
     \frac{\partial F}{\partial t} (\gamma(s),t(s)) = \nabla r (F(\gamma(s),t(s))),
 \end{equation*}
 Hence the second term on the RHS of \eqref{flow} only adds a contribution to the $r$ component of $\dot\sigma(s)$ (but not to the angular components). We can compute $DF$ by solving the ODE
 \begin{equation*}
     \frac{\partial}{\partial t} DF(x,t) = D(\nabla r)(F(x,t)) DF(x,t)
 \end{equation*}
 with initial condition $DF(x,0) = \operatorname{Id}$. In Cartesian coordinates $D\nabla r$ takes a block diagonal form:
 \begin{equation*}
      D\nabla r_{ij} = \begin{cases}
      a_i \delta_{ij} &\textrm{for $i,j=1,...,\lambda $}, \\
      \delta_{ij} &\textrm{for $i,j=\lambda+1,...,n $}, \\
      0 &\textrm{otherwise},
      \end{cases}
 \end{equation*}
 hence
 \begin{equation} \label{DF}
     DF(x,t)_{ij} = \begin{cases}
      e^{a_i t} \delta_{ij} &\textrm{for $i,j=1,...,\lambda $}, \\
      \delta_{ij} &\textrm{for $i,j=\lambda+1,...,n $}, \\
      0 &\textrm{otherwise}.
      \end{cases}
 \end{equation}
 Moreover, we have that $DF(x,t) \partial_r \propto \partial_r$ because $F$ is the flow of a vector field collinear to $\partial r$. Therefore
 \begin{equation*}
     \pi_\Theta DF \ V = \pi_\Theta DF \pi_\Theta V \textrm{ for any } V \in TM,
 \end{equation*}
 and thus we can write 
 \begin{align*}
     h(\pi_\Theta \dot{\sigma},\pi_\Theta  \dot{\sigma}) &= h(\pi_\Theta DF \pi_\Theta \dot{\gamma},\pi_\Theta DF \pi_\Theta  \dot{\gamma}) \\
     &\leq h(DF \pi_\Theta \dot{\gamma},DF \pi_\Theta  \dot{\gamma}) \\
     &\leq h(\pi_\Theta \dot{\gamma},\pi_\Theta  \dot{\gamma}).
 \end{align*}
 Here we have first used that the orthogonal projection $\pi_\Theta$ cannot increase the norm, and then that $DF$ cannot increase the norm either, because it does not increase any of the Cartesian components \eqref{DF}.
 
 \textbf{Step 2} ($\rrho$ direction)\textbf{.} We want to show that $\frac{\Vert \nabla f \Vert^2}{\Vert \nabla \rrho \Vert^2}$ does not increase when shifting $r^2$ and $\rho^2$ by $\epsilon$, so that we do not get a larger contribution in \eqref{metricrrho}. Thus in what follows we view $r$ and $\rho$ as functions of $\epsilon$, in the sense that $r^2 = r_0^2 + \epsilon$ and $\rho^2 = \rho_0^2 + \epsilon$ with respect to some reference values $r_0, \rho_0$ (but we will omit the subscript $0$ from the notation). From this point of view, what we want to show is
 \begin{equation*}
     \frac{\partial}{\partial \epsilon} \Big\vert_{\epsilon = 0} \frac{\Vert \nabla f \Vert^2}{\Vert \nabla \rrho \Vert^2} \geq 0.
 \end{equation*}
 We begin with some preliminary computations, where $\nu := 2 - \frac{1}{p}$, and all derivatives are evaluated at $\epsilon = 0$.
 \begin{align*}
     %\frac{\partial}{\partial \epsilon} \rho &= \frac{1}{2\rho}, \\
     \frac{\partial}{\partial \epsilon} \rho^2 &= 1, \\
     \frac{\partial}{\partial \epsilon} \rho^4 &= 2 \rho^2, \\
     %\frac{\partial}{\partial \epsilon} r \rho &= \frac{1}{2} (\frac{r}{\rho} + \frac{\rho}{r}), \\
     \frac{\partial}{\partial \epsilon} (r \rho)^{2\nu} &= \nu (r \rho)^{2\nu-2} (r^2 + \rho^2), \\
     \Vert \nabla f \Vert^2 &= \Vert d f \Vert^2 = \Vert d(r^2) \Vert^2 + \Vert d(\rho^2) \Vert^2, \\
     \Vert \nabla \rrho \Vert^2 &= \Vert d \rrho \Vert^2 = \frac{1}{(2p(r\rho)^\nu)^2} \left(\rho^4 \Vert d(r^2) \Vert^2 + r^4 \Vert d(\rho^2) \Vert^2\right).
 \end{align*}
 Moreover, we need the following estimates, where $a = 2 \min_{i = 1,...,n} a_i$, $A = 2 \max_{i = 1,...,n} a_i$.
 \begin{align*}
      \mina \rho^2 \leq \ &\Vert d(\rho^2) \Vert^2 \leq \maxa \rho^2, \\
      \mina \leq \frac{\partial}{\partial \epsilon} &\Vert d(\rho^2) \Vert^2 \leq \maxa.
 \end{align*}
 These are easily proven in Cartesian coordinates. 
 
 Applying the chain rule and substituting the previous computations and estimates, we get, after a lengthy but trivial computation, the estimate
 \begin{align*}
     %(2p)^{2}(r\rho)^{4\nu} \Vert \nabla \rrho \Vert^4
     \frac{\partial}{\partial \epsilon} \frac{\Vert \nabla f \Vert^2}{\Vert \nabla \rrho \Vert^2} %=\ &(r\rho)^{4\nu} \Vert \nabla \rrho \Vert^2 \frac{\partial}{\partial \epsilon} (\rho r)^{2\nu} \frac{\Vert d(r^2) \Vert^2 + \Vert d(\rho^2) \Vert^2}{\rho^4 \Vert d(r^2) \Vert^2 + r^4 \Vert d(\rho^2) \Vert^2} \\
     %\geq\ & \big( \nu (\rho r)^{2\nu - 2} (r^2 + \rho^2) (ar^2 + a\rho^2) (a \rho^4 r^2 + a \rho^2 r^4) \\
      %&+ (\rho r)^{2\nu} 2a (a \rho^4 r^2 + a \rho^2 r^4) \\
      %&- (\rho r)^{2\nu} (A r^2 + A \rho^2)(A \rho^4 + 2 A \rho^2 r^2 + A r^4 + 2 A \rho^2 r^2) \big) \\
     %=\ &(\rho r)^{2\nu} \big( \nu a^2(r^2 + \rho^2)^3 + 2a^2 ( \rho^4 r^2 +  \rho^2 r^4) \\
     %&- A^2(r^2 +  \rho^2)(\rho^4 + 4 \rho^2 r^2 + r^4) \big) \\
     %=\ & (\rho r)^{2\nu} \big( (\nu a^2 - A^2)(r^6 +\rho^6) \\ &+ ((3\nu + 2) a^2 - 5 A^2)(\rho^4 r^2 + \rho^2 r^4) \big).
     \geq \frac{(\nu a^2 - A^2)(r^6 +\rho^6) + \left((3\nu + 2) a^2 - 5 A^2\right)(\rho^4 r^2 + \rho^2 r^4)}{(2p(r\rho)^\nu)^2 \Vert \nabla \rrho \Vert^4},
 \end{align*}
 where the RHS is guaranteed to be positive if
 \begin{equation*}
     \frac{A^2}{a^2} \leq \min \{ \nu, \frac{3 \nu + 2}{5} \}.
 \end{equation*}
 Since $\nu \in (1,2)$ only enters in our choice of coordinates, we can freely choose it. In particular, as long as
 \begin{equation*}
     \frac{A^2}{a^2} \leq \frac{8}{5},
 \end{equation*}
 we can choose $\nu$ close enough to $2$ so that $\frac{\partial}{\partial \epsilon} \frac{\Vert \nabla f \Vert^2}{\Vert \nabla \rrho \Vert^2} \geq 0$.
 
 \textbf{Step 3.} (Final argument)\textbf{.} It simply remains to compare $g(\dot\sigma,\dot\sigma)$ to $g(\dot\gamma,\dot\gamma)$, term by term, according to \eqref{metricrrho}.
 By step 1, we have
 \begin{align*}
     h_\Theta(\dot\sigma,\dot\sigma) &\leq h_\Theta(\dot\gamma,\dot\gamma), \\
     h_\Phi(\dot\sigma,\dot\sigma) &\leq h_\Phi(\dot\gamma,\dot\gamma).
 \end{align*}
  Moreover, by computing $\Vert df \Vert^2$ in Cartesian coordinates, and using the fact that under the flow of $\nabla r$ and $- \nabla \rho$, the Cartesian coordinates are non-increasing (in absolute value), it is easy to check that
  \begin{equation*}
      \Vert df \Vert^2 (\sigma(s)) \leq \Vert df \Vert^2 (\gamma(s)).
  \end{equation*}
  By step 2, we have that
  \begin{equation*}
      \frac{\Vert \nabla f \Vert^2}{\Vert \nabla \rrho \Vert^2}(\sigma(s)) \leq \frac{\Vert \nabla f \Vert^2}{\Vert \nabla \rrho \Vert^2}(\gamma(s)),
  \end{equation*}
  and since the $\partial_\rrho$ component of $\dot\sigma$ is the same as that of $\dot\gamma$ (because $\rrho_{\gamma(s)}$ and $\rrho_{\sigma(s)}$ only differ by a constant),
  \begin{equation*}
      d\rrho(\dot\sigma) = d\rrho(\dot\gamma).
  \end{equation*}
  Finally, because $f_{\gamma(s)} = f_{\sigma(s)}$, we have
 \begin{equation*}
     df(\dot\sigma) = df(\dot\gamma).
 \end{equation*}
 Plugging all of the above into \eqref{metricrrho}, we conclude that
 \begin{equation*}
     g(\dot\sigma,\dot\sigma) \leq g(\dot\gamma,\dot\gamma),
 \end{equation*}
 as desired.
\end{proof}

The following lemma is an easy consequence of Lemma \ref{lemtip}, and from it we can derive Lemma \ref{lemnbhd}.

\begin{lem} \label{lemtip2}
  Suppose Condition \ref{condmain} is satisfied, and let $q \in \Imi(\pt)$. Then there exists a point $\tilde{q} \in \JpM(q)$ such that $\rrho_{\tilde{q}}=0$ and $f_{\tilde{q}}<0$. Equivalently, $\rho_{\tilde{q}}=0$ and $r_{\tilde{q}}>0$.
\end{lem}

\begin{proof}
 The equivalence of the two statements follows simply by definition \eqref{frrho}. Now for the proof of existence: If $\rrho_q = 0$, choose $\tilde{q} = q$, and we are done because $q \in \Imi(\pt)$ implies $f_q < f_{\pt} = 0$. Otherwise, choose a f.d.\ timelike curve $\gamma \colon [0,1] \to \Mwp$ from $\gamma(0) = q$ to $\gamma(1) = \pt$, and write it in components as in \eqref{gammaincoos}. If $\rrho(1/2)=0$, choose $\tilde{q} = \gamma(1/2)$, noting that $\gamma(1/2) \in \Imi(\pt)$ and therefore $f_{\gamma(1/2)}<f_{\pt} = 0$. If $\rrho(1/2)\neq 0$, then since $\gamma(1/2) \in \IpM(q)$, we can choose $0 < \varepsilon < \rrho(1/2)$ small enough so that $\hat{q} := (f(1/2),\rrho(1/2)-\varepsilon,\Theta(1/2),\Phi(1/2)) \in \IpM(q)$. Then, by Lemma \ref{lemtip}, there exists a f.d.\ timelike curve $\sigma$ from $\hat{q}$ until some point $\tilde{q} := \sigma(s_\epsilon)$ such that $\rrho_{\tilde{q}}=0$. Moreover, $f_{\tilde{q}}=f(s_\epsilon)<f(1)=0$.
\end{proof}

\begin{proof}[Proof of Lemma \ref{lemnbhd}]
 Let $q \in \Imi(\pt)$. Then we can choose $\tilde{q} \in \JpM(q)$ as in Lemma \ref{lemtip2}. We claim that $\tilde{q} \in \Imi(\pt)$, not only with respect to our metric $g$, but even with respect to the metric $\giso$ (see Lemma \ref{lemgiso}). To prove this claim, simply note that we can reach $\pt$ from $\tilde{q}$ by following the integral curve of $\nabla f$ through $\tilde{q}$ (which has $\rho=0$ initially, hence $\rho=0$ on the whole integral curve, while $r$ must decrease, thus we reach $\pt$). By Lemma \ref{lemapp}, $\tilde{q} \in \Imi(\pt,\giso)$ implies that there exists a neighborhood $U$ of $\pt$ such that $U \setminus \{ \pt \} \subseteq \IpM(\tilde{q},\giso)$. By Lemma \ref{lemgiso}, $\IpM(\tilde{q},\giso) \subseteq \IpM(\tilde{q})$, and since $\tilde{q} \in \JpM(q)$, it follows that $U \setminus \{ \pt \} \subseteq \IpM(q)$, as desired.
\end{proof}

\subsection{Limit curves, push-up and proof of Theorem \ref{mainthm}} \label{secpushup}

Having established the crucial Lemma \ref{lemnbhd}, the rest of the proof of Theorem \ref{mainthm} does not require any computations in coordinates. Yet it follows the same philosophy of showing that the causal relation on $\Mwp$ has some of the same (good) properties that it would have on a non-degenerate spacetime.

The next lemma is a sort of limit curve theorem, but can also be interpreted as telling us that $(\Mwp,h,f,\zeta)$ is causally simple (see \cite[p.\ 65]{BEE} for causal simplicity of non-degenerate spacetimes).

\begin{lem}[$\overline{\Ip} \subseteq \Jp$] \label{limcurv}
 Suppose Condition \ref{condopen} is satisfied. Let $(p_i)_i$, $(q_i)_i$ be sequences of points in $\Mwp$ such that $q_i \in \Ip(p_i)$. If $p_i \to p$ and $q_i \to q$, then $q \in \Jp(p)$.
\end{lem}

\begin{proof}
 We first show the case $p_i,q_i,p,q \neq \pt$. Then, by Lemma \ref{avoidp}, there exists a sequence of f.d.\ timelike curves $\gamma_i \colon [a_i,b_i] \to  \Mwop$ such that $\gamma_i(a_i) = p_i$ and $\gamma_i(b_i) = q_i$. The idea is quite simple: we claim that $(\gamma_i)$, up to a subsequence, converges to a causal curve $\gamma\colon [a,b] \to \Mwp$. We show this by applying the usual limit curve theorem \cite[Thm.\ 3.1]{Min} on the spacetime $\Mwop$. It is necessary to distinguish between the case when the limit curve is also in $\Mwop$, and the case when the limit curve crosses over the singular point $\pt$ (then, technically speaking, there are two limit curves in $\Mwop$, which can be joined in $\Mwp$).
 \begin{itemize}
  \item \textbf{Case 1:} The sequence $\gamma_i$ converges uniformly to a causal curve $\gamma\colon [a,b] \to \Mwop$, or to a single point. Either way, $q \in \JpM(p) \subseteq \Jp(p)$, and we are done.
  \item \textbf{Case 2:} There exist reparametrizations $\gamma^p_i \colon [0,b^p_i) \to \Mwop$ of $\gamma_i$ and a future endless (in $\Mwop$) causal curve $\gamma^p \colon [0,\infty) \to M$ with $\gamma(0)=p$ such that $\gamma^p_i \to \gamma^p$ uniformly on compact subsets. Analogously, there exist reparametrizations $\gamma^q_i \colon (-b^q_i,0] \to \Mwop$ of $\gamma_i$ and a past endless causal curve $\gamma^q \colon (-\infty,0] \to M$ with $\gamma^q(0)=y$ such that $\gamma^q_i \to \gamma^q$ uniformly on compact subsets.
 \end{itemize}
 In case 2, we claim that $\lim_{t \to \infty} \gamma^p(t) = \lim_{s \to -\infty} \gamma^q(s) = \pt$. This is a direct consequence of Lemma \ref{lempbdry} and the fact that $f$ is bounded away from $0,1$ on $\gamma^p$ and $\gamma^q$, hence $\gamma^p, \gamma^q$ cannot run into the boundary $\partial \Mwp$. But if $\lim_{t \to \infty} \gamma^p(t) = \lim_{s \to -\infty} \gamma^q(s) = \pt$, then (after suitable reparametrization) we can extend $\gamma^p, \gamma^q$ to $\pt$ and concatenate them, forming a causal curve in $\Mwp$ that joins $p$ with $q$, as desired.
 
 In case that some of $p_i,q_i,p,q$ equal $\pt$, we can proceed with an analogous proof, but we have to add a third case, where $\pt$ is an endpoint of the limit curve.
\end{proof}

The next lemma is well-known for non-degenerate spacetimes.

\begin{lem}[Push-up] \label{pushup}
 Suppose Condition \ref{condopen} is satisfied. If $q \in \Jp(p)$, then $\Ip(q) \subseteq \Ip(p)$.
\end{lem}

\begin{proof} %picture possible
 If $p = q$, the result is trivial, so assume $p \neq q$, and let $\sigma \colon [0,1] \to \Mwp$ be a causal curve from $p$ to $q$.
 
 \textbf{Case 1: $q = \pt$.} Let $\tilde{q} \in \IpM(\pt)$ (see Figure \ref{figcase1}). Then, by Condition \ref{condopen}, there exists a neighborhood $U$ of $\pt$ such that $U \setminus \{ \pt \} \subseteq \ImM(\tilde{q})$. Because $\sigma(1) = \pt$, there must exist some $0 < s_0 < 1$ such that $\sigma(s_0) \in U \setminus \{ \pt \}$. But then $\tilde{q} \in \IpM(\sigma(s_0))$, and since also $\sigma(s_0) \in \JpM(p)$, we conclude by the standard push-up lemma in $\Mwop$ that $\tilde{q} \in \IpM(p) \subseteq \Ip(p)$. Since $\tilde{q}$ was arbitrary, we are finished with this case.
 
 \textbf{Case 2: $p = \pt$.} The argument is similar to the one in \cite[Prop.~2.1]{AGH}. Let $\tilde{q} \in \Ip(q) = \IpM(q)$ (see Figure \ref{figcase2}). We construct a timelike curve $\gamma$ from $\pt$ to $\tilde{q}$. Let $y_n := \sigma(1/n)$, and choose a point $z_1 \in \IpM (y_1) \cap \ImM(\tilde{q})$. By openness of $\ImM(z_1)$, and since, by the usual push-up lemma in $M$, $y_2 \in \JmM(y_1) \subset \ImM(z_1)$, we may choose $z_2 \in \ImM(z_1) \cap \IpM(y_2) \cap B^h_{1/2}(y_2)$. Here $B^h_{1/2}(y_2)$ denotes the ball of radius $1/2$ around $y_2$, measured with respect to the Riemannian metric $h$. Iterating this procedure, we obtain a sequence $(z_l)_l$ such that $z_l \in \ImM(z_{l-1}) \cap \IpM(y_l) \cap B^h_{1/l}(y_l)$. Then we construct $\gamma$ by joining all the timelike segments going from $z_l$ to $z_{l+1}$. Since, by contruction, $\lim_{l \to \infty} z_l = \lim_{l \to \infty} y_l = \pt$, the timelike curve $\gamma$ connects $\pt$ and $\tilde{q}$.
 
 \textbf{Case 3: $p,q \neq \pt$.} If $\sigma$ lies entirely in $\Mwop$, the result follows from the standard theory. Therefore, we assume w.l.o.g.\ that $\sigma(\frac{1}{2}) = \pt$. Then, in particular, $q \in \Jp(\pt)$, so by case 2, we have that $\Ip(q) \subseteq \Ip(\pt)$. But since also $\pt \in \Jp(p)$, by case 1 it follows that $\Ip(\pt) \subseteq \Ip(p)$, and we are done.
\end{proof}

\begin{figure}
\centering
  \begin{subfigure}[b]{0.49\textwidth}
    \centering
    \begin{tikzpicture}[scale=1.5]
     \fill[blue!25] (0,0) ellipse (1.3 and 0.8);
     \draw[thick,red] (-2,-2) -- (0,0);
     \filldraw[black] (0,2) circle (1.3pt) node[anchor=south] {$\tilde{q}$};
     \filldraw[black] (0,0) circle (1.3pt) node[anchor=west] {$\sigma(1) = \pt$};
     \filldraw[black] (-2,-2) circle (1.3pt) node[anchor=west] {$\sigma(0) = p$};
     \filldraw[black] (-0.5,-0.5) circle (1.3pt) node[anchor=west] {$\sigma(s_0)$};
     \draw  plot[smooth, tension=.7] coordinates {(0,0) (0.2,0.5) (-0.1,1.4) (0,2)};
     \draw[-{Stealth[scale=1.1]}] (0.1995,0.49) -- (0.2,0.5);
     \draw  plot[smooth, tension=.7] coordinates {(-0.5,-0.5) (-0.1,0.5) (-0.2,1.5) (0,2)};
     \draw[-{Stealth[scale=1.1]}] (-0.102,0.49) -- (-0.1,0.5);
     \node at (-0.7,0.3) {$U$};
    \end{tikzpicture}
  \caption{Case 1.} \label{figcase1}
  \end{subfigure} \hfill
  \begin{subfigure}[b]{0.49\textwidth}
   \centering
   \begin{tikzpicture}[scale=1.5]
    \draw[thick,red] (-2,-2) -- (0,0);
    \filldraw[black] (0,2) circle (1.3pt) node[anchor=south] {$\tilde{q}$};
    \filldraw[black] (0,1) circle (1.3pt) node[anchor=south east] {$z_1$};
    \filldraw[black] (0,0) circle (1.3pt) node[anchor=north west] {$y_1 = q$};
    \filldraw[black] (-1,-0.5) circle (1.3pt) node[anchor=south east] {$z_2$};
    \filldraw[black] (-1,-1) circle (1.3pt) node[anchor=north west] {$y_2$};
    \filldraw[black] (-1.5,-1.25) circle (1.3pt) node[anchor=south east] {$z_3$};
    \filldraw[black] (-1.5,-1.5) circle (1.3pt) node[anchor=north west] {$y_3$};
    \filldraw[black] (-2,-2) circle (1.3pt) node[anchor=north west] {$\pt$};
    \draw  plot[smooth, tension=1] coordinates {(0,2) (0,1) (-1,-0.5) (-1.5,-1.25) (-2,-2)};
    \node at (-0.3,-0.7) {\textcolor{red}{$\sigma$}};
    \node at (-0.7,0.3) {$\gamma$};
   \end{tikzpicture}
   \caption{Case 2.}  \label{figcase2}
  \end{subfigure}
  \caption{An illustration of the proof of Lemma \ref{pushup}. The red line represents $\sigma$, and the black curves represent future-directed causal curves.} 
\end{figure}
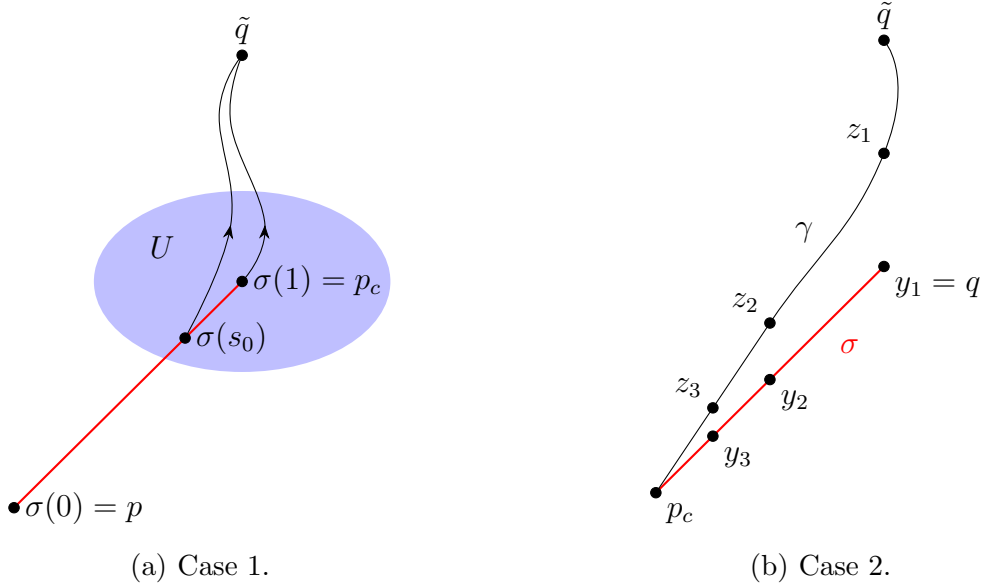

Lemma \ref{lemfinal} below, together with Lemma \ref{lemnbhd}, completes the proof of Theorem \ref{mainthm}.

\begin{lem} \label{lemfinal}
 If Condition \ref{condopen} is satisfied, then the Morse spacetime $(\Mwop,g)$ is causally continuous.
\end{lem}

\begin{proof}
  By Definition \ref{defcc} (in Appendix \ref{appcc}), $(\Mwop,g)$ is causally continuous if it is distinguishing and reflecting. Because $f$ is a time function, $(\Mwop,g)$ must be distinguishing \cite[Sec.\ 2]{BDGSS}. Thus we only need to prove reflectivity. Let $p,q \in \Mwop$ be such that $\ImM(p) \subseteq \ImM(q)$ (the future case is analogous). We need to prove that $\IpM(q) \subseteq \IpM(p)$. By the time-reverse of Lemma \ref{avoidp}, $\Imi(p) \subseteq \Imi(q)$, and then, since $p \in \overline{\Imi(p)} \subseteq \overline{\Imi(q)}$, Lemma \ref{limcurv} tells us that $p \in \Jm(q)$. But then, by Lemma \ref{pushup}, $\Ip(q) \subseteq \Ip(p)$, which again by Lemma \ref{avoidp} implies $\IpM(q) \subseteq \IpM(p)$.
\end{proof}

\begin{rem}
 One may even say that the Morse geometry $(\Mwp,h,f,\zeta)$ is globally hyperbolic. Firstly, it follows from Lemma \ref{lempbdry} that $f$ is a Cacuhy time function, in the sense that any causal curve that is inextendible in $\Mwp$, must start at one boundary component and end at the other, crossing each level set exactly once. Secondly, by compactness of $\Mwp$, it is easy to see that the causal diamonds $\Jp(p) \cap \Jm(q)$ are compact, for all $p,q \in \Mwp$. However, both arguments are also true when $(M,g)$ is causally discountinuous, such as in the index $1,n-1$ case. The point we would like to make here, is that one should additionally require Condition \ref{condopen} to hold, and then the causal structure of $\Mwp$ is very well-behaved.
 
 We can turn this remark into a mathematically precise statement by using the language of Lorentzian length spaces, introduced by Kunzinger and S\"amann \cite{KuSa} (see also \cite{BuGH}, where topology change is discussed in this context). Lorentzian length spaces are topological spaces equipped with a notion of causal order and satisfying a set of axioms which, in particular, imply a version of Condition \ref{condopen} \cite[Lem.\ 2.12]{KuSa}. A somewhat related point is that $(\Mwop,g)$ is semi-globally hyperbolic, meaning that it can be divided into globally hyperbolic pieces, which in our case are separated by the critical level sets of $f$. The notion of semi-globally hyperbolic spacetime was introduced by Janssen in \cite{Jan}, with the goal of defining quantum field theories on them (note the connection to Conjecture \ref{Sconj}).
\end{rem}

\section{Towards a full resolution of the Borde--Sorkin Conjecture} \label{secfull}

Throughout this section, we employ the same notational conventions as in Section \ref{secproof}, except that we allow our Morse functions to have multiple critical points. The current progress on Conjecture \ref{BSconj} is summarized in Theorem \ref{thm2} in the Introduction. What still remains open is the case when $f$ has critical points of index $\lambda = 2,...,n-2$, and $h$ is arbitrary. In other words, we do not know what happens if we drop Condition \ref{condmain}.

Notice that Condition \ref{condmain} is basically telling us two things:
\begin{enumerate}
    \item We can find a coordinate neighborhood $\Uwp$ of $\pt$ where both $h$ and $f$ take a specified standard form.
    \item We have the bounds $\zeta_c < \zeta$ and $\zeta_c \leq 8/5$ (see \eqref{largezeta}), which can be interpreted as a bound on how much anisotropy is allowed.
\end{enumerate}
In the first part of this section, we show that the neighborhood $\Uwp$ can always be found, the only difference being that in the general case, we need to add a perturbation to $h$ that vanishes at $\pt$. In the second part of this section, we give a candidate counterexample to Conjecture \ref{BSconj}, which suggests that $\zeta > \zeta_c$ is a necessary condition for causal continuity. We then conclude by proposing a modified version of the conjecture which takes this into account.

\subsection{Generalized standard neighborhoods} \label{secgennbhd}

The statement of the next proposition can be seen as a weaker version of the first part of Condition \ref{condmain}.

\begin{prop} \label{proploc}
 Let $(\Mwp,h,f,\zeta)$ be a Morse geometry and $\pt$ be a critical point of $f$. Then there exists an open neighborhood $\Uwp \subseteq \Mwp$ of $\pt$, an open ball $\mathcal{B} \in \real^n$ around the origin, and a coordinate chart $\phi : \Uwp \to \mathcal{B}$ such that $\phi(\pt) = 0$ and
 \begin{align*}
    &f \circ \phi^{-1} = \frac{1}{2} \sum_{i} a_i (x^i)^2,
    &h \circ \phi^{-1} = \sum_{i} (dx^i)^2 + \sum_{k,l} \hpert_{k l}(x) dx^k dx^l,
\end{align*}
for some real constants $a_i \neq 0$ and some tensor $\hpert$ satisfying $\hpert(\pt)=0$.
\end{prop}

The proof relies on the following two results from the literature.

\begin{simdiag}
 Let $H$, $D$ be two real symmetric $n \times n$ matrices, and let $H$ be positive definite. Then there exists a real non-singular matrix $\Lambda$ such that both $\Lambda^T H \Lambda$ and $\Lambda^T D \Lambda$ are diagonal.
\end{simdiag}

\begin{Morselem}
 Let $\Mwp$ be a manifold, $f \colon \Mwp \to \real$ be a Morse function and $\pt$ be a critical point of $f$. Then there exists an open neighborhood $\Uwp \subseteq \Mwp$ of $\pt$, an open ball $\mathcal{B} \subseteq \real^n$ around the origin, and a coordinate chart $\phi : \Uwp \to \mathcal{B}$ such that $\phi(\pt) = 0$ and
 \begin{equation*}
    f \circ \phi^{-1} = \frac{1}{2} \sum_{i} a_i (x^i)^2,
 \end{equation*}
 for some real constants $a_i \neq 0$.
\end{Morselem}

\begin{proof}[Proof of Proposition \ref{proploc}]
By the Morse Lemma, we can find coordinates where $f$ already has the desired form. Then we apply a linear change of coordinates in order to simultaneously diagonalize the bilinear forms $h(\pt)$ and $\hess f(\pt)$. Finally, we scale each coordinate, in order to normalize our new basis with respect to $h(\pt)$.
\end{proof}

We use Proposition \ref{proploc} to formulate a relaxed version of Condition \ref{condmain}. Let $(\Mwp,h,f,\zeta)$ be a Morse geometry, and suppose we have chosen a critical point $\pt$.

\begin{cond} \label{condrelax}
 The constants $a_i$ appearing in Proposition \ref{proploc} (applied to $\pt$) satisfy
 \begin{equation*}
    \zeta > \zeta_c := \max_{i,j} \left\vert \frac{a_i}{a_j} \right\vert.
 \end{equation*}
\end{cond}

We do not include the bound $\zeta_c \leq 8/5$, because it will not be relevant in the upcoming examples, and it seems likely to not be a necessary condition for causal continuity. Note also that in this paper, we take the point of view that $\zeta$ is specified as part of the Morse geometry. If, instead, we only specify $h$ and $f$, then we can always choose $\zeta > \zeta_c$ at a given critical point (hence also at any finite number of critical points). In this sense, Condition \ref{condrelax} is not very restrictive. Note, in any case, that $\zeta_c$ depends on both $h$ and $f$ in a neighborhood of $\pt$.

Since the coordinate system that we get from Proposition \ref{proploc} is not necessarily unique, the question arises whether the truth or falsehood of Condition \ref{condrelax} depends on any coordinate choices. The answer is no. To see this, note that $\max_i a_i$ is the maximum value of the quadratic form $\hess f (\pt)$ applied to the $h(\pt)$-unit ball. Similarly, $\min_i a_i$ is the minimum. These maxima and minima are independent of the basis, so we conclude that the value of $\zeta_c$ is the same among all coordinate bases satisfying the properties listed in Proposition \ref{proploc}.

One is left to wonder how much our proofs in Section \ref{secproof} are affected when adding the perturbation $\hpert$ that appears in Proposition \ref{proploc}. We can say the following:
\begin{itemize}
    \item If Condition \ref{condrelax} is satisfied, then a generalization of Lemma \ref{lemgiso} holds. The only difference is that in \eqref{giso}, we need to replace $\zeta$ by $1 < \hat{\zeta} < \zeta$. This makes the lightcones of $\giso$ a bit narrower, compensating for the fact that the perturbation $\hpert$ might have also made the lightcones of $g$ a bit narrower.
    \item The proofs in Section \ref{secopen} are no longer valid.
    \item If we can prove that Condition \ref{condrelax} implies Condition \ref{condopen} (compare with Lemma \ref{lemnbhd}), then causal continuity follows by the same arguments as in Section \ref{secpushup}.
    \item On general spacetimes, causal continuity is not stable under perturbations of the metric. This remains true even if we require said perturbations to always widen or narrow the lightcones with respect to the original metric (see Examples \ref{ex3} and \ref{ex4} in Appendix \ref{appcc}).
\end{itemize}

\subsection{A potential counterexample} \label{seccounter}

The discussion in Section \ref{secgennbhd} leads to a natural question: is Condition \ref{condrelax} necessary in order to have causal continuity? If the answer is yes, it would mean that Conjecture \ref{BSconj} is false in its original form. We believe that this is indeed so. Constructing examples that violate Condition \ref{condrelax} is easy, but showing that they are causally discontinuous is not (and we do not achieve it in this paper).

The following examples are meant to illustrate what happens when Condition \ref{condrelax} is not satisfied. For convenience, we take $\Mwp$ non-compact and without boundary, the idea being that it represents a neighborhood of a critical point in some larger Morse geometry.

\begin{ex} \label{ex1}
 Let $\Mwp$ be an open ball in $\real^2$, centered around the origin $\pt := 0$, and equipped with coordinates $(x, y)$. Define
\begin{align*}
 h &:= dx^2 + dy^2, 
 & f:= - \frac{1}{2}\left( x^2 + \pb y^2 \right),
\end{align*}
so that
\[
g = (x^2 + \pb^2 y^2)(dx^2+dy^2) - \zeta (x dx + \pb y dy)^2.
\]
The case $\pb = 1$ is considered in Appendix \ref{appnbhd}: it is a \emph{neighborhood spacetime} in the sense of Borde et al. \cite{BDGSS}. We will refer to is as an \emph{isotropic} neighborhood, while in the case of $\pb \neq 1$, we will talk about \emph{anisotropic} neighborhoods.

Consider a radial line $\gamma(s) = (s,ms)$, for $m \in \real$. Then
\begin{equation} \label{radialeq}
 g\left(\dot{\gamma}(s),\dot{\gamma}(s)\right) = \left( (1-\zeta)\pb^2 m^4 + (1 + \pb^2 -2\zeta \pb) m^2 + 1 - \zeta \right) s^2.
\end{equation}
If $\pb = 1$, this quantity reduces to
\begin{equation*}
   g\left(\dot{\gamma}(s),\dot{\gamma}(s)\right) = \left(1-\zeta\right)\left(m^2+1\right)^2 s^2, 
\end{equation*}
which negative for all $m$, hence all radial lines are timelike (see Figure \ref{figex1a}). In this case, the past of $\pt$, which is the whole Morse spacetime $\Mwop = \Mwp \setminus \{ \pt \}$, can be written as a single TIP, $\Mwop = \Imi(\pt) = \ImM(\gamma)$, for any future directed timelike curve $\gamma$ that ends at $\pt$ (see Appendix \ref{appnbhd}). 

Taking the limit $\pb \to 0$ in \eqref{radialeq}, we get an expression that is positive whenever $m^2 > 1-\zeta$. By continuity, we conclude that for $0 < \pb < 1$ small enough, $g(\dot{\gamma}(s),\dot{\gamma}(s))$ can be negative, zero or positive, depending on $m$ (the dependence, however, is more complicated than in the $\pb \to 0$ limit). Concretely, for $\pb$ small (relative to $\zeta$):
\begin{itemize}
 \item There exist null radial lines, which form the boundaries of the future sets\footnote{Here by \emph{future set} we mean a set $\mathcal{F}$ such that $\IpM(\mathcal{F}) = \mathcal{F}$. Analogously, a \emph{past set} $\mathcal{P}$ satisfies $\ImM(\mathcal{P}) = \mathcal{P}$.} $\mathcal{F}_1,\mathcal{F}_2$ and the pasts sets $\mathcal{P}_1,\mathcal{P}_2$. Intuitively, this happens because the lightcones tilt much faster when moving in the $x$ direction, compared to moving in the $y$ direction (see Figure \ref{figex1b}). Regardless, we still have that $M = \Imi(\pt)$.
 \item $\Ip (q)$ is not open, for any $q \in \Mwop$. This is because $\pt \in \Ip (q)$, but no neighborhood of $\pt$ is entirely contained in $\Ip(q)$ (compare with Lemma \ref{lemnbhd}).
\end{itemize}
Note that there also exists an intermediate case, when $\pb$ has the exact value so that setting \eqref{radialeq} equal to zero has degenerate solutions, and then the boundaries of the sets $\mathcal{F}_1,\mathcal{F}_2,\mathcal{P}_1,\mathcal{P}_2$ overlap. The case of $\pb > 1$ very large can be reduced to the case of $\pb < 1$ small by rescaling both $x$ and $y$ by a factor of $1/\pb$ (then also $h$ is rescaled, but this does not affect $g$).

Since the critical point in this example has index $0$, the Morse spacetime $(M,g)$ is causally continuous, no matter how small we choose $\pb$ (see Theorem \ref{thm2}). In fact, $(M,g)$ is even globally hyperbolic, with $f$ being a Cauchy time function. Note however that in this example, $\Ip(\pt) = \emptyset$, which simplifies things a lot.
\end{ex}

 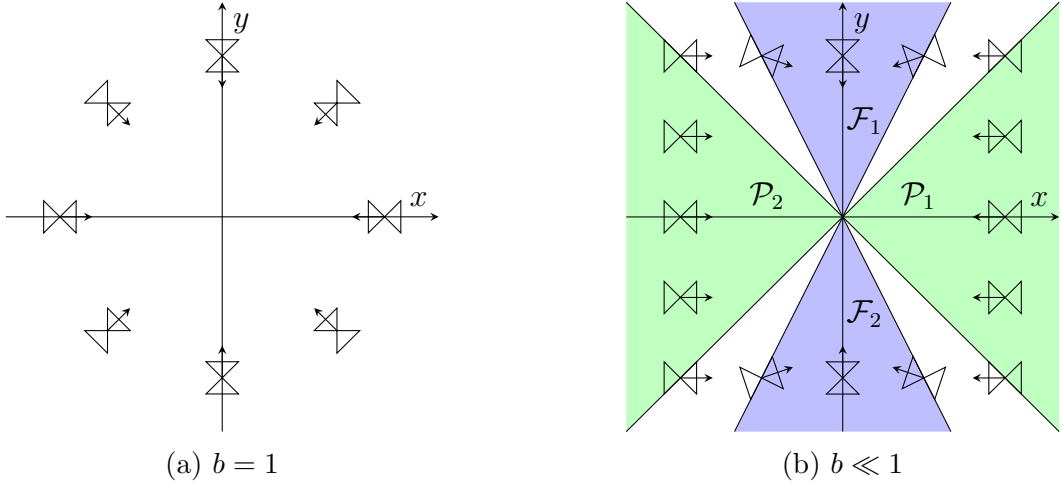
\begin{figure}
 \centering
 \begin{subfigure}[b]{0.49\textwidth}
 \centering
 \begin{tikzpicture}
   \begin{axis}[xmin=-2,xmax=2,
                 ymin=-2,ymax=2,
                 axis equal image,
                 axis on top=true,
                 axis x line=middle,
                 axis y line=middle,
                 xlabel={$x$},
                 ylabel={$y$},
                 xtick=\empty,
                 ytick=\empty,
                 clip=false]
        \draw [-stealth](axis cs: 1.06,1.06) -- (axis cs: 0.85,0.85);
        \draw (axis cs: 1.06,0.85) -- (axis cs: 0.85,1.06) -- (axis cs: 1.27,1.06) -- (axis cs: 1.06,1.27) -- cycle;
        \draw [-stealth](axis cs: -1.06,1.06) -- (axis cs: -0.85,0.85);
        \draw (axis cs: -1.06,0.85) -- (axis cs: -0.85,1.06) -- (axis cs: -1.27,1.06) -- (axis cs: -1.06,1.27) -- cycle;
        \draw [-stealth](axis cs: 1.06,-1.06) -- (axis cs: 0.85,-0.85);
        \draw (axis cs: 1.06,-0.85) -- (axis cs: 0.85,-1.06) -- (axis cs: 1.27,-1.06) -- (axis cs: 1.06,-1.27) -- cycle;
        \draw [-stealth](axis cs: -1.06,-1.06) -- (axis cs: -0.85,-0.85);
        \draw (axis cs: -1.06,-0.85) -- (axis cs: -0.85,-1.06) -- (axis cs: -1.27,-1.06) -- (axis cs: -1.06,-1.27) -- cycle;
        \draw [-stealth](axis cs: 1.5,0) -- (axis cs: 1.2,0);
        \draw (axis cs: 1.65,0.15) -- (axis cs: 1.65,-0.15) -- (axis cs: 1.35,0.15) -- (axis cs: 1.35,-0.15) -- cycle;
        \draw [-stealth](axis cs: -1.5,0) -- (axis cs: -1.2,0);
        \draw (axis cs: -1.65,0.15) -- (axis cs: -1.65,-0.15) -- (axis cs: -1.35,0.15) -- (axis cs: -1.35,-0.15) -- cycle;
        \draw [stealth-](axis cs: 0,-1.2) -- (axis cs: 0,-1.5);
        \draw (axis cs: 0.15,-1.65) -- (axis cs: -0.15,-1.65) -- (axis cs: 0.15,-1.35) -- (axis cs: -0.15,-1.35) -- cycle;
        \draw [stealth-](axis cs: 0,1.2) -- (axis cs: 0,1.5);
        \draw (axis cs: 0.15,1.65) -- (axis cs: -0.15,1.65) -- (axis cs: 0.15,1.35) -- (axis cs: -0.15,1.35) -- cycle;
     \end{axis}
   \end{tikzpicture}
   \caption{$\pb = 1$}  \label{figex1a}
 \end{subfigure} \hfill
 \begin{subfigure}[b]{0.49\textwidth}
 \centering
   \begin{tikzpicture}
    \begin{axis}[xmin=-2,xmax=2,
                 ymin=-2,ymax=2,
                 axis equal image,
                 axis on top=true,
                 axis x line=middle,
                 axis y line=middle,
                 xlabel={$x$},
                 ylabel={$y$},
                 xtick=\empty,
                 ytick=\empty,
                 clip=false]
        \draw[name path=a] (axis cs:0,0) -- (axis cs: 2,2);
        \draw[name path=b] (axis cs:0,0) -- (axis cs: 2,-2);
        \draw[name path=c] (axis cs:0,0) -- (axis cs: -2,2);
        \draw[name path=d] (axis cs:0,0) -- (axis cs: -2,-2);
        \draw[name path=e] (axis cs:0,0) -- (axis cs: 1,2);
        \draw[name path=f] (axis cs:0,0) -- (axis cs: 1,-2);
        \draw[name path=g] (axis cs:0,0) -- (axis cs: -1,2);
        \draw[name path=h] (axis cs:0,0) -- (axis cs: -1,-2);
        \path[name path=line1] (axis cs: -1,2) -- (axis cs: 0,2);
        \path[name path=line2] (axis cs: 0,2) -- (axis cs: 1,2);
        \path[name path=line3] (axis cs: -1,-2) -- (axis cs: 0,-2);
        \path[name path=line4] (axis cs: 0,-2) -- (axis cs: 1,-2);
        \addplot[green!25] fill between[of = a and b];
        \addplot[green!25] fill between[of = c and d];
        \addplot[blue!25] fill between[of = e and line2];
        \addplot[blue!25] fill between[of = g and line1];
        \addplot[blue!25] fill between[of = h and line3];
        \addplot[blue!25] fill between[of = f and line4];
        \draw [-stealth](axis cs: 1.5,1.5) -- (axis cs: 1.2,1.5);
        \draw [-stealth](axis cs: 1.5,-1.5) -- (axis cs: 1.2,-1.5);
        \draw [-stealth](axis cs: -1.5,1.5) -- (axis cs: -1.2,1.5);
        \draw [-stealth](axis cs: -1.5,-1.5) -- (axis cs: -1.2,-1.5);
        \draw (axis cs: 1.65,1.65) -- (axis cs: 1.65,1.35) -- (axis cs: 1.35,1.65) -- (axis cs: 1.35,1.35) -- cycle;
        \draw (axis cs: -1.65,1.65) -- (axis cs: -1.65,1.35) -- (axis cs: -1.35,1.65) -- (axis cs: -1.35,1.35) -- cycle;
        \draw (axis cs: 1.65,-1.65) -- (axis cs: 1.65,-1.35) -- (axis cs: 1.35,-1.65) -- (axis cs: 1.35,-1.35) -- cycle;
        \draw (axis cs: -1.65,-1.65) -- (axis cs: -1.65,-1.35) -- (axis cs: -1.35,-1.65) -- (axis cs: -1.35,-1.35) -- cycle;
        \draw [-stealth](axis cs: 0.75,1.5) -- (axis cs: 0.46,1.4);
        \draw (axis cs: 0.55,1.6) -- (axis cs: 0.65,1.3) -- (axis cs: 0.85,1.7) -- (axis cs: 0.95,1.4) -- cycle;
        \draw [-stealth](axis cs: -0.75,1.5) -- (axis cs: -0.46,1.4);
        \draw (axis cs: -0.55,1.6) -- (axis cs: -0.65,1.3) -- (axis cs: -0.85,1.7) -- (axis cs: -0.95,1.4) -- cycle;
        \draw [-stealth](axis cs: 0.75,-1.5) -- (axis cs: 0.46,-1.4);
        \draw (axis cs: 0.55,-1.6) -- (axis cs: 0.65,-1.3) -- (axis cs: 0.85,-1.7) -- (axis cs: 0.95,-1.4) -- cycle;
        \draw [-stealth](axis cs: -0.75,-1.5) -- (axis cs: -0.46,-1.4);
        \draw (axis cs: -0.55,-1.6) -- (axis cs: -0.65,-1.3) -- (axis cs: -0.85,-1.7) -- (axis cs: -0.95,-1.4) -- cycle;
        \draw [-stealth](axis cs: 1.5,0) -- (axis cs: 1.2,0);
        \draw (axis cs: 1.65,0.15) -- (axis cs: 1.65,-0.15) -- (axis cs: 1.35,0.15) -- (axis cs: 1.35,-0.15) -- cycle;
        \draw [-stealth](axis cs: -1.5,0) -- (axis cs: -1.2,0);
        \draw (axis cs: -1.65,0.15) -- (axis cs: -1.65,-0.15) -- (axis cs: -1.35,0.15) -- (axis cs: -1.35,-0.15) -- cycle;
        \draw [stealth-](axis cs: 0,-1.2) -- (axis cs: 0,-1.5);
        \draw (axis cs: 0.15,-1.65) -- (axis cs: -0.15,-1.65) -- (axis cs: 0.15,-1.35) -- (axis cs: -0.15,-1.35) -- cycle;
        \draw [stealth-](axis cs: 0,1.2) -- (axis cs: 0,1.5);
        \draw (axis cs: 0.15,1.65) -- (axis cs: -0.15,1.65) -- (axis cs: 0.15,1.35) -- (axis cs: -0.15,1.35) -- cycle;
        \draw [-stealth](axis cs: 1.5,0.75) -- (axis cs: 1.2,0.75);
        \draw (axis cs: 1.65,0.9) -- (axis cs: 1.65,0.6) -- (axis cs: 1.35,0.9) -- (axis cs: 1.35,0.6) -- cycle;
        \draw [-stealth](axis cs: -1.5,0.75) -- (axis cs: -1.2,0.75);
        \draw (axis cs: -1.65,0.9) -- (axis cs: -1.65,0.6) -- (axis cs: -1.35,0.9) -- (axis cs: -1.35,0.6) -- cycle;
        \draw [-stealth](axis cs: 1.5,-0.75) -- (axis cs: 1.2,-0.75);
        \draw (axis cs: 1.65,-0.9) -- (axis cs: 1.65,-0.6) -- (axis cs: 1.35,-0.9) -- (axis cs: 1.35,-0.6) -- cycle;
        \draw [-stealth](axis cs: -1.5,-0.75) -- (axis cs: -1.2,-0.75);
        \draw (axis cs: -1.65,-0.9) -- (axis cs: -1.65,-0.6) -- (axis cs: -1.35,-0.9) -- (axis cs: -1.35,-0.6) -- cycle;
        \node at (axis cs:0.7,0.2) {$\mathcal{P}_1$};
        \node at (axis cs:-0.7,0.2) {$\mathcal{P}_2$};
        \node at (axis cs:0.2,0.9) {$\mathcal{F}_1$};
        \node at (axis cs:0.2,-0.9) {$\mathcal{F}_2$};
     \end{axis}
   \end{tikzpicture}
  \caption{$\pb \ll 1$} \label{figex1b}
  \end{subfigure}
  \caption{The causal structure of Example \ref{ex1}.}  \label{figex1}
\end{figure}

Building upon the previous example, we propose our candidate counterexample to Conjecture \ref{BSconj}.

\begin{ex} \label{ex2}
Let $\Mwp$ be an open ball in $\real^4$, centered around the origin $\pt := 0$, and equipped with coordinates $(x, y, z, w)$. Define
\begin{align*}
 h &:= dx^2 + dy^2 + dz^2 + dw^2, 
 & f:= \frac{1}{2}\left( -x^2 - \pb y^2 + z^2 + w^2 \right),
\end{align*}
and $g$, as usual, by \eqref{metriccoo}. We claim that reflectivity (see Definition \ref{defcc}) is violated by the pair of points
\begin{align*}
 p &:= (1,0,0,0),
 & q:= (m,0,1,0),
\end{align*}
 where $m := \frac{\sqrt{\zeta}-1}{\sqrt{\zeta-1}}$. From the causal analysis of the punctured $(x,z)$-plane (see Figure \ref{figex2} and Appendix \ref{appnbhd}) it follows that $\IpM(q) \subseteq \IpM(p)$. However, $\ImXZ(q) \not\supseteq \ImXZ(p)$, where the subscript $XZ$ means that we are only considering causal curves in the punctured $(x,z)$-plane. Nonetheless, it is possible that $\ImM(q) \supseteq \ImM(p)$ when also considering causal curves that leave said plane. In particular, if $\pb = 1$ (or close enough to $1$), we see from the analysis in Example \ref{ex1} (see also Figure \ref{figapp1a}) that from $p$ we can reach the negative $x$-axis via a future-directed timelike curve contained in the $(x,y)$-plane. Then, from the negative $x$-axis, we can reach $q$. If $\pb$ is too small, however, this construction is no longer possible (see Figure \ref{figex1b}), suggesting that probably $\ImM(q) \not\supseteq \ImM(p)$. This is not a bulletproof argument, of course, because we are ignoring all timelike curves that are not contained in any coordinate plane.
\end{ex}
  
\begin{figure}
  \begin{center}

   \begin{tikzpicture}
    \begin{axis}[xmin=-2,xmax=2,
                 ymin=-2,ymax=2,
                 axis equal image,
                 axis background/.style={fill=blue!25},
                 axis on top=true,
                 axis x line=middle,
                 axis y line=middle,
                 xlabel={$x$},
                 ylabel={$z$},
                 xtick=\empty,
                 ytick=\empty,
                 clip=false]
     %\filldraw[black] (axis cs: 0,0) circle (2pt) node[anchor=north] {$\pt$};
     \filldraw[black] (axis cs: 1,0) circle (2pt);
     \node at (axis cs:1,-0.3) {$p$};
     \filldraw[black] (axis cs: -0.577,1) circle (2pt) node[anchor=north east] {$q$};
     \draw (axis cs:0,0) -- (axis cs: -2*0.577,2);
     \draw[name path=futureofpc] (axis cs:0,0) -- (axis cs: -2*0.577,-2);
     \draw[name path=pastofpc] (axis cs:0,0) -- (axis cs: -2,-2*0.577);
     \draw[dashed] (axis cs:0,0) -- (axis cs: 2,2*0.577);
     \draw[dashed] (axis cs:0,0) -- (axis cs: 2,-2*0.577);
     \addplot[domain=-2:-0.577,samples=20,name path=futureofp] {0.577*x+2*0.577*sqrt(x^2+1)}; %past q
     \addplot[domain=-0.577:0.84,samples=30,name path=pastofp] {0.577*x+2*0.577*sqrt(x^2+1)}; %future q
     \addplot[domain=0:0.925,samples=25, name path=pastofp2]({0.577*x+0.577*sqrt(4*x^2+3)},{x}); %past p
     \addplot[domain=0:0.925,samples=25, name path=pastofp1]({0.577*x+0.577*sqrt(4*x^2+3)},{-x}); %past p
     \addplot[domain=0:2,samples=50,name path=futureofp2]({-0.577*x+0.577*sqrt(4*x^2+3)},{x}); %future p
     \addplot[domain=0:2,samples=50,name path=futureofp1]({-0.577*x+0.577*sqrt(4*x^2+3)},{-x}); %future p
     \path[name path=xaxis1] (axis cs: -2,0) -- (axis cs: 0,0);
     \path[name path=xaxis2] (axis cs: -0.577,0) -- (axis cs: 0,0);
     \path[name path=xaxis3] (axis cs: 0,0) -- (axis cs: 1,0);
     \path[name path=xaxis4] (axis cs: 1,0) -- (axis cs: 2,0);
     \path[name path=line1] (axis cs: -2,2) -- (axis cs: -2*0.577,2);
     \path[name path=line2] (axis cs: -2*0.577,2) -- (axis cs: 0.84,2);
     \path[name path=line3] (axis cs: 1.358,2) -- (axis cs: 2,2);
     \path[name path=line4] (axis cs: 1.358,-2) -- (axis cs: 2,-2);
     \path[name path=line5] (axis cs: -2,-2) -- (axis cs: -2*0.577,-2);
     \addplot[white] fill between[of = futureofp2 and pastofp2];
     \addplot[white] fill between[of = line3 and pastofp2];
     \addplot[white] fill between[of = futureofp1 and pastofp1];
     \addplot[white] fill between[of = line4 and pastofp1];
     \addplot[green!25] fill between[of = pastofpc and futureofp];
     \addplot[white] fill between[of = line1 and futureofp];
     \addplot[white] fill between[of = line5 and pastofpc];
     \addplot[blue!50] fill between[of = line2 and pastofp];
     \addplot[green!25] fill between[of = pastofp1 and pastofp2];
     \node at (axis cs:-1.4,0.4) {$\ImXZ(q)$};
     \node at (axis cs:-0.5,1.7) {$\IpXZ(q)$};
     %\node at (axis cs:-0.4,1.8) {$I^+(p)$};
     \node at (axis cs:1.8,-0.3) {$\ImXZ(p)$};
     \node at (axis cs:0.5,-1.2) {$\IpXZ(p)$};
     \end{axis}
   \end{tikzpicture}
  \end{center}
  \caption{The points $p,q$ of Example \ref{ex2} and their futures and pasts (restricted to the $(x,z)$-plane $XZ$).} \label{figex2}
 \end{figure}
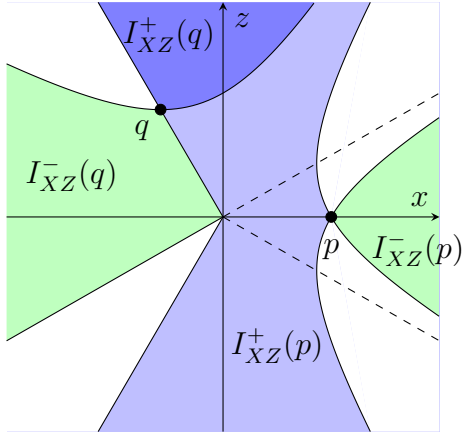

\section{Conclusions} \label{secconc}

With the proof of Theorem \ref{mainthm}, we have established a new case of the Borde--Sorkin conjecture (Conjecture \ref{BSconj} in the Introduction). In doing so, we have advanced the current progress on the conjecture to that summarized in Theorem \ref{thm2}. Along the way, we have developed a notion of causal structure for Morse geometries that includes the critical points. This supports the view that degenerate metrics are physically reasonable, and is a first step towards understanding quantum fields on Morse geometries (in view of Conjecture \ref{Sconj}).

Let us briefly mention here three recent approaches to quantum field theory that are specially relevant for topology change. The algebraic approach of Janssen \cite{Jan} has been developed specifically with topology change as one of its applications, but the existence of states in this approach is still an open problem. Another approach is that of Sorkin and Johnston \cite{Joh,Sor3}; so far it has been applied to the trousers topology change \cite{BDJS} (confirming the energy divergences there), but, to our knowledge, not to any of the causally continuous Morse geometries (where Conjecture \ref{Sconj} predicts a well-behaved QFT). Lastly, the recent paper of Kontsevich and Segal \cite{KoSe} defines QFTs on the category of cobordisms with certain complex metrics. Real Lorentzian metrics arise as a limit case of these complex metrics, and so do Morse metrics \cite{LoSo,Wit}. So far, however, it is only known that a QFT is induced on the limit spacetime when the latter is a non-degenerate and globally hyperbolic \cite[Thm.\ 5.2]{KoSe}. It remains to be seen if this result can be generalized to Morse geometries.

Another important conclusion of the present paper is that we have found a potential counterexample to Conjecture \ref{BSconj} (Example \ref{ex2} in Section \ref{seccounter}). Despite being of dimension $4$ and having only a critical point of index $2$, we believe that our example is causally discontinuous, due to being highly anisotropic (i.e.\ because it has $\zeta_c > \zeta$). The lack of symmetries and good coordinate choices has prevented us from proving this fully. Regardless, we propose the following refinement of Conjecture \ref{BSconj}, which incorporates a bound on the anisotropy (Condition \ref{condrelax}).

\begin{conj} \label{BSconjmod}
 Let $(\Mwp,h,f,\zeta)$ be a Morse geometry of dimension $n$, and $(\Mwop,g)$ the corresponding Morse spacetime. Assume $f$ has a single critical point for each critical value. Then $(M,g)$ is causally continuous if and only if the following hold:
 \begin{enumerate}
     \item None of the critical points has index $1$ or $n-1$,%If $f$ has at least one critical point of index $\lambda = 1, n-1$, or one of index $\lambda \neq 0,n$ where Condition \ref{condrelax} is violated, then $(\Mwop,g)$ is causally discontinuous.
     \item Condition \ref{condrelax} is satisfied at every critical point of index different from $0,n$.%If each critical point of $f$ has index $\lambda=0,n$, or has any index $\lambda \neq 1, n-1$ and satisfies Condition \ref{condrelax}, then $(\Mwop,g)$ is causally continuous.
 \end{enumerate}
\end{conj}

In order to prove Conjecture \ref{BSconjmod}, two steps remain. One is to show that Condition \ref{condrelax} is really necessary, by showing that Example \ref{ex2} (where it is violated) is causally discontinuous. The other remaining step is to generalize Theorem \ref{mainthm} by adding a perturbation to $h$ that vanishes at the critical points, and by removing the requirement that $\zeta_c \leq 8/5$. This is nontrivial, because causal continuity is, in general, not stable under perturbations (see Examples \ref{ex3} and \ref{ex4} in Appendix \ref{appcc}). Yet the second half of our proof (Section \ref{secpushup}) is robust under perturbations, and does not require $\zeta_c \leq 8/5$, so it would suffice to prove openness of the chronological relation in $\Mwp$ (Condition \ref{condopen}), and then causal continuity would follow.

\section*{Acknowledgements}

I would like to thank Elefterios Soultanis and Maximilian Ruep for very interesting discussions, and Annegret Burtscher for comments on the draft. I am also grateful to Bernardo Araneda and Simon Pepin Lehalleur for pointing me to reference \cite{KoSe}.

\appendix

\section{Neighborhood spacetimes} \label{appnbhd}

In this appendix, we review the causal structure of isotropic neighborhood spacetimes, as studied in Borde et al.\ \cite{BDGSS}. At the end, we prove Lemma \ref{lemapp}, which is new, although it follows quite straightforwardly from the analysis in  \cite{BDGSS}.

Let $\Mwp \subseteq \real^n$ be an open neighborhood of the origin, equipped with a coordinate system $(x^1, ..., x^\lambda, y^1, ..., y^{n - \lambda})$, where $\lambda \neq 0,1,n-1,n$, and 
\begin{align} \label{nbhdeqs}
 h &:= \sum_{i=1}^{\lambda} (dx^i)^2 + \sum_{j=1}^{n-\lambda} (dy^j)^2, 
 & f:= -\frac{1}{2} \sum_{i=1}^{\lambda} (x^i)^2 + \frac{1}{2} \sum_{j=1}^{n-\lambda} (y^j)^2.
\end{align}
Note that the origin $\pt = (0,...,0)$, is the only critical point of $f$ in this case. We write $\Mwop := \Mwp \setminus \{ \pt \}$, as usual (here $\Mwp$ has no boundary, but can be thought of as a neighborhood of a critical point in some cobordism). It is convenient to change to polar coordinates $(\rho,\Theta,r,\Phi)$, where
\begin{align*}
 &\rho := \sum_{j=1}^{n-\lambda} (y^j)^2,
 &r := \sum_{i=1}^{\lambda} (x^i)^2,
\end{align*}
and where $\Theta, \Phi$ denote the angular coordinates corresponding to the subspaces $\{r = 0  \}$ and $\{\rho = 0 \}$ respectively (thus each of $\Theta, \Phi$ is a collection of angular variables, rather than a single one). The Lorentzian metric \eqref{metricinv} is then given by
\begin{align}
\begin{split}
 g = &\left(r^2-(\zeta-1) \rho^2 \right)d\rho^2 +\left(\rho^2-(\zeta-1) r^2 \right)dr^2 + 2 \zeta \rho r d\rho dr \\
 &+ \left(\rho^2 + r^2 \right) \left(\rho^2 d\Theta^2 + r^2 d\Phi^2 \right).
 \label{nbhdg}
\end{split}
\end{align}
Because the coefficients in front of $d \Theta$ and $d \Phi$ are positive, any null geodesic in the $\Theta, \Phi = \text{const.}$ plane (with respect to the restricted metric), must also be a null geodesic in the full spacetime.

Thus we start by commenting on the situation with constant angles. Recall also that in 2 dimensions, any null curve is automatically a null geodesic. The implicit equations for any null geodesic can thus be found from \eqref{nbhdg}:
\begin{equation} \label{eqapp}
 \sqrt{\zeta - 1} (x^2 - y^2 ) = \pm 2xy + c_\pm,
\end{equation}
where $c_\pm$ are constants. In particular, the case of $c_\pm=0$ corresponds to geodesics that bound $\Ipm(\pt)$, which are radial lines of a certain slope (depicted as dashed lines in Figure \ref{figapp1b}). We can use this information to find the future and past sets of any point (depicted as colored regions in Figure \ref{figapp1b}).

Next we analyse the case when of $\rho =0$ and $\Theta = \text{const.}$ (the case $r =0$ and $\Phi = \text{const.}$ is analogous). We restrict to the case where $\Phi = \phi$ is just a single angular variable, hence again reducing the problem to two dimensions. The null geodesics on the $\rho =0, \Theta = \text{const.}$ plane (with respect to the induced metric) are then given by
\[
 r(\phi)=r_0 \operatorname{e}^{\pm \phi / \sqrt{\zeta-1} }.
\]
Again, we can use this to find the future and past sets of any point on the plane, with respect to the induced metric (see Figure \ref{figapp1a}). In this case, there are no null geodesics going through the origin $\pt$, and all points with $\rho =0$ lie in the past of $\pt$.

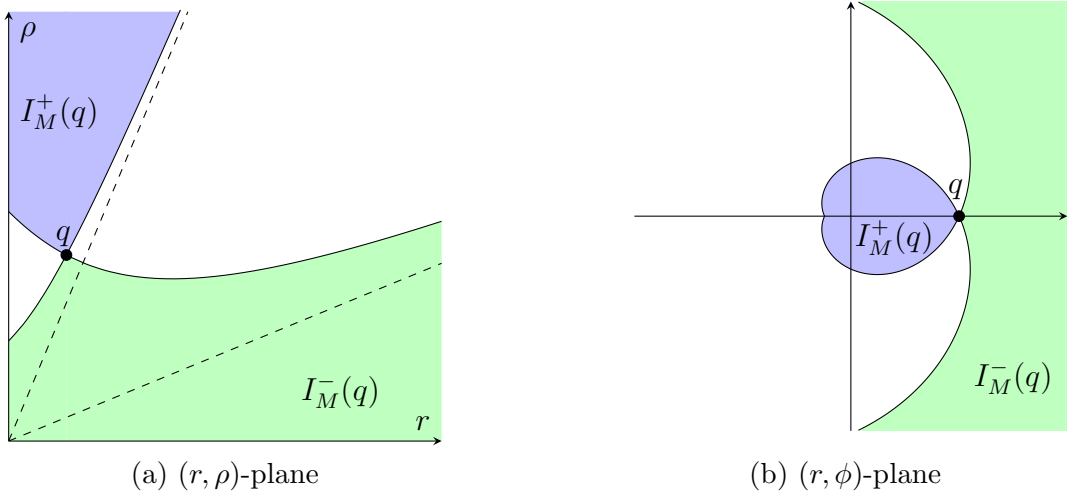
\begin{figure}
 \centering
 \begin{subfigure}[b]{0.49\textwidth}
 \centering
   \begin{tikzpicture}
    \begin{axis}[xmin=0,xmax=3,
                 ymin=0,ymax=3,
                 axis equal image,
                 axis on top=true,
                 axis x line=middle,
                 axis y line=middle,
                 xlabel={$r$},
                 ylabel={$\rho$},
                 xtick=\empty,
                 ytick=\empty,
                 clip=false]
     \filldraw[black] (axis cs: 0.4,1.3) circle (2pt);
     \node at (axis cs:0.38,1.43) {$q$};
     \draw[dashed] (axis cs:0,0) -- (axis cs: 3,3*0.414);
     \draw[dashed] (axis cs:0,0) -- (axis cs: 3*0.414,3);
     \addplot[domain=0:0.402,samples=10,name path=future1,smooth] {-x+sqrt(2*x^2+2.57)};
     \addplot[domain=0.4:3,samples=20,name path=past2,smooth] {-x+sqrt(2*x^2+2.57)};
     \addplot[domain=0:0.4,samples=5,name path=past1,smooth] {x+sqrt(2*x^2+0.49)};
     \addplot[domain=0.4:1.19,samples=10,name path=future2,smooth] {x+sqrt(2*x^2+0.49)};
     \path[name path = bottom1] (axis cs:0,0)--(axis cs:0.4,0);
     \path[name path = bottom2] (axis cs:0.4,0)--(axis cs:3,0);
     \path[name path = top1] (axis cs:0,3)--(axis cs:0.402,3);
     \path[name path = top2] (axis cs:0.39,3)--(axis cs:3*0.399,3);
     \addplot[green!25] fill between[of = bottom1 and past1];
     \addplot[green!25] fill between[of = bottom2 and past2];
     \addplot[blue!25] fill between[of = future1 and top1];
     \addplot[blue!25] fill between[of = future2 and top2];
     \draw (axis cs: 2.3, 0.35) node {$\ImM(q)$};
     \draw (axis cs: 0.35, 2.3) node {$\IpM(q)$};
     \end{axis}
   \end{tikzpicture}
  \caption{$(r,\rho)$-plane} \label{figapp1b}
  \end{subfigure} \hfill
 \begin{subfigure}[b]{0.49\textwidth}
 \centering
 \begin{tikzpicture}
   \begin{axis}[xmin=-2, xmax=2, ymin=-2, ymax=2, xtick=\empty, ytick = \empty,  xticklabels={}, xlabel={}, ylabel={}, yticklabels={},axis background/.style={fill=green!25},axis on top=true,axis equal image,
                 axis x line=middle,
                 axis y line=middle,clip=false]
        \addplot[domain=0:88,samples=88,smooth,name path = A,data cs=polar] (-x,{exp(x*pi/400)});
        \addplot[domain=0:88,samples=88,smooth,name path = B,data cs=polar] (x,{exp(x*pi/400)});
        \addplot[domain=0:180,samples=180,smooth,name path = C,data cs=polar] (-x,{exp(-x*pi/400)});
        \addplot[domain=0:180,samples=180,smooth,name path = D,data cs=polar] (x,{exp(-x*pi/400)});
        \path[name path = E] (axis cs:-2.1,2.1)--(axis cs:0.1,2.1);
        \path[name path = F] (axis cs:-2.1,-2.1)--(axis cs:0.1,-2.1);
        \addplot[white] fill between[of = A and B];
        \addplot[white] fill between[of = E and F];
        \addplot[blue!25] fill between[of = C and D];
        \filldraw[black] (axis cs:1,0) circle (2pt);
        \draw (axis cs: 1.5, -1.5) node {$\ImM(q)$};
        \draw (axis cs: 0.4, -0.2) node {$\IpM(q)$};
        \draw (axis cs: 0.97,0.25) node {$q$};
   \end{axis}
   \end{tikzpicture}
   \caption{$(r,\phi)$-plane}  \label{figapp1a}
 \end{subfigure}
  \caption{The causal structure of an isotropic neighborhood.}
\end{figure}

\begin{lem} \label{lemapp}
  Let $(\Mwp,h,f,\zeta)$ be an isotropic Morse neighborhood, with $h,f$ given by \eqref{nbhdeqs}. For every $q \in \Ipm(\pt)$, there exists a neighborhood $U$ of $\pt$ such that $U \setminus \{\pt\} \subseteq \ImpM(q)$.
\end{lem}

\begin{proof}
 We show the case when $q \in \Imi(\pt)$. W.l.o.g.\ we may choose our coordinates such that $q = (\rho_q,0,r_q,0)$, where necessarily $\rho_q \neq 0$ (otherwise $q$ cannot be in $\Imi(\pt)$). We want to find $\rho_0, r_0$ such that all points $(\rho,\Theta,r,\Phi) \neq \pt$ with $\rho < \rho_0$, $r < r_0$ and $\Theta, \Phi$ arbitrary, are contained in $\IpM(q)$. By symmetry, we may choose our coordinates such that at most one of the $\Theta$- and one of the $\Phi$-angles may be different from zero, hence effectively reducing the problem to four dimensions.
 
 Our argument now resembles the one in the proof of \cite[Claim 1]{BDGSS}. Let $\ll$ denote the chronological relation in $\Mwop$. If $\rho_q, r_q \neq 0$, then
 \[
  x = (\rho_x,0,r_x,0) \ll (\epsilon_1,0,0,0)  \ll (\epsilon_1 \delta,\theta,0,0) \ll (\epsilon_1 \delta \epsilon_2,\theta,\epsilon_3,\phi).
 \]
 In every step where we have added an $\epsilon$, we have used our analysis of the causal structure in the case $\theta, \phi$ constant. In the step where we have added $\delta$, it is using our analysis of the $\rho =0$ and $\theta = \text{const.}$ case. In principle, $\delta$ depends on $\theta$. We see, however, that the ``worst case scenario'' (when $\delta$ has to be the smallest) is when $\theta = \pi$. Thus we can choose this largest value, so that the procedure works in all cases. Note also that in the last step, since we starting from the origin of the $(r,\phi)$-plane, we can choose any value for $\phi$ that we want. Setting $\rho_0 := \epsilon_1 \delta \epsilon_2$ and $r_0 := \epsilon_3$, and again considering the causal structure in the case $\theta, \phi$ constant, we are done.
\end{proof}

\section{Causal continuity} \label{appcc}

Let $(M,g)$ be a non-degenerate spacetime. We refer to \cite[Chap.\ 3]{BEE} for the basic concepts and notation of causality theory. The idea is that $(M,g)$ is causally continuous if the set valued functions $q \mapsto \IpmM(q)$ are continuous. There are various equivalent ways to make this precise \cite[pp.\ 59-71]{BEE}. In this paper, we use the following definition, which is perhaps the most standard one, even though it does not directly capture the intuition behind the concept.

\begin{defn} \label{defcc}
  A spacetime $(\Mwop,g)$ is called
  \begin{enumerate}
      \item \emph{distinguishing} if
      \begin{equation*}
      \ImM(p) = \ImM(q) \iff p=q \iff \IpM(p) = \IpM(q)
      \end{equation*}
      for all $p,q \in M$,
      \item \emph{reflecting} if
      \begin{equation*}
      \ImM(p) \subseteq \ImM(q) \iff \IpM(p) \supseteq \IpM(q)
      \end{equation*}
      for all $p,q \in M$,
      \item \emph{causally continuous} if it is distinguishing and reflecting.
  \end{enumerate}
\end{defn}

The following example shows that causal continuity is not stable under perturbations of the metric $g$, even if we only allow perturbations that make the lightcones narrower.

\begin{ex} \label{ex3}
 Let $M := \real^2 \setminus \{(x,t) \mid x \geq 2 \vert t \vert \}$ and $g_\alpha = -\alpha dt^2 + dx^2$. Then $(M,g_\alpha)$ is causally continuous for $\alpha \geq 2$ and causally discontinuous for $\alpha < 2$. This can be seen in Figure \ref{figapp2a}: Reflectivity is violated for pairs of points lying on the diagonal red line, one above and one below the origin (such as the depicted points $p,q$). The red line has slope $1/\alpha$, hence if $\alpha \geq 2$, half of the red line lies inside the removed wedge, and there is no violation of reflectivity anymore.
\end{ex}

The next example shows that causal continuity is not stable under widening of the lightcones, either.

\begin{ex} \label{ex4}
Let $M := \real^2 \setminus \{(x,t) \mid t \leq - 2 \vert x \vert \}$ and $g_\alpha = -\alpha dt^2 + dx^2$. Then $(M,g_\alpha)$ is causally continuous for $\alpha \leq \frac{1}{2}$ and causally discontinuous for $\alpha > \frac{1}{2}$. The argument is similar to the one in Example \ref{ex3} (see Figure \ref{figapp2b}).
\end{ex}

\begin{figure}
 \centering
 \begin{subfigure}[b]{0.49\textwidth}
 \centering
 \begin{tikzpicture}
   \begin{axis}[xmin=-2, xmax=2, ymin=-2, ymax=2, hide axis,axis equal image,clip=false]
        \draw[name path=wedgeupper] (axis cs: 0,0) -- (axis cs: 2,1);
        \draw[name path=wedgelower] (axis cs: 0,0) -- (axis cs: 2,-1);
        \draw[red,thick] (axis cs: -2,-2) -- (axis cs: 2,2);
        \filldraw[black] (axis cs: -1,-1) circle (2pt); 
        \node at (axis cs: -1,-1.2) {$p$}; %p
        \filldraw[black] (axis cs: 1,1) circle (2pt); 
        \node at (axis cs: 1,0.8) {$q$}; %q
        \path[name path = top] (axis cs: -2,2) -- (axis cs: 2,2);
        \path[name path = top2] (axis cs: 0,2) -- (axis cs: 2,2);
        \path[name path = bottom] (axis cs: -2,-2) -- (axis cs: 0,-2);
        \draw[name path = pastp] (axis cs: 0,-2) -- (axis cs: -1,-1);
        \addplot[green!25] fill between[of=bottom and pastp];
        \draw[name path = pastq] (axis cs: 1,1) -- (axis cs: 1.33,0.67);
        \draw[name path = futureq] (axis cs: 1,1) -- (axis cs: 0,2);
        \draw[name path = futurep] (axis cs: -1,-1) -- (axis cs: -2,0);
        \addplot[green!25] fill between[of=wedgeupper and pastq];
        \addplot[blue!25] fill between[of=futurep and top];
        \addplot[blue!50] fill between[of=futureq and top2];
        \addplot[gray!50] fill between[of = wedgelower and wedgeupper];
        \node at (axis cs: 1.5,0) {remove};
        \node at (axis cs: 1,1.7) {$\IpF(q)$};
        \node at (axis cs: -1,1) {$\IpF(p)$};
        \node at (axis cs: -1,-1.7) {$\ImF(p)$};
        \node[anchor = north] at (axis cs: 1,-1) {$\ImF(q)$};
        \draw[-stealth]  (axis cs: 1,-1) -- (axis cs: 0.9,0.6);
   \end{axis}
   \end{tikzpicture}
   \caption{Example \ref{ex3} with $\alpha = 1$.}  \label{figapp2a}
 \end{subfigure} \hfill
  \begin{subfigure}[b]{0.49\textwidth}
 \centering
   \begin{tikzpicture}
    \begin{axis}[xmin=-2, xmax=2, ymin=-2, ymax=2, hide axis,axis equal image,clip=false]
        \draw[name path=wedgeleft] (axis cs: 0,0) -- (axis cs: -1,-2);
        \draw[name path=wedgeright] (axis cs: 0,0) -- (axis cs: 1,-2);
        \draw[red,thick] (axis cs: -2,-2) -- (axis cs: 2,2);
        \filldraw[black] (axis cs: -1,-1) circle (2pt); 
        \node at (axis cs: -1,-1.2) {$p$}; %p
        \filldraw[black] (axis cs: 1,1) circle (2pt); 
        \node at (axis cs: 1,0.8) {$q$}; %q
        \path[name path = top] (axis cs: -2,2) -- (axis cs: 2,2);
        \path[name path = top2] (axis cs: 0,2) -- (axis cs: 2,2);
        \path[name path = bottom] (axis cs: -2,-2) -- (axis cs: 0,-2);
        \path[name path = bottom2] (axis cs: -1,-2) -- (axis cs: 1,-2);
        \path[name path = bottom3] (axis cs: 1,-2) -- (axis cs: 2,-2);
        \draw[name path = pastp] (axis cs: -1,-1) -- (axis cs: -0.67,-1.33);
        \addplot[green!25] fill between[of=bottom and pastp];
        \draw[name path = pastq] (axis cs: 1,1) -- (axis cs: 2,0);
        \draw[name path = futureq] (axis cs: 1,1) -- (axis cs: 0,2);
        \draw[name path = futurep] (axis cs: -1,-1) -- (axis cs: -2,0);
        \addplot[green!25] fill between[of=wedgeright and pastq];
        \addplot[green!25] fill between[of=bottom3 and pastq];
        \addplot[blue!25] fill between[of=futurep and top];
        \addplot[blue!50] fill between[of=futureq and top2];
        \addplot[gray!50] fill between[of = bottom2 and wedgeright];
        \node at (axis cs: 0,-1.8) {remove};
        \node at (axis cs: 1,1.7) {$\IpF(q)$};
        \node at (axis cs: -1,1) {$\IpF(p)$};
        \node at (axis cs: -1.28,-1.75) {$\ImF(p)$};
        \node at (axis cs: 1,0) {$\ImF(q)$};
        %\draw[-stealth]  (axis cs: 1,-1) -- (axis cs: 0.9,0.6);
   \end{axis}
   \end{tikzpicture}
  \caption{Example \ref{ex4} with $\alpha = 1$.} \label{figapp2b}
  \end{subfigure}
  \caption{A pair of points $p,q$ for which reflectivity is violated.}
\end{figure}
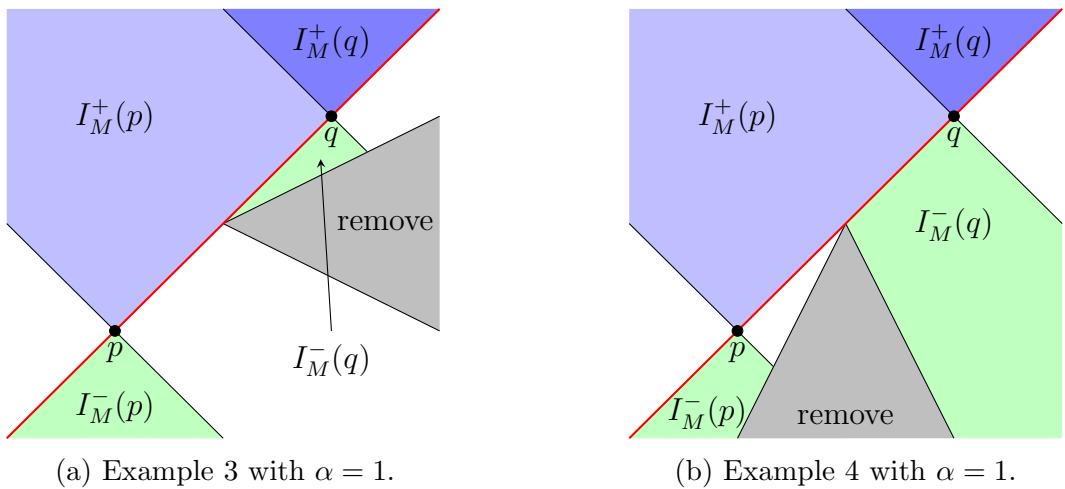

\bibliographystyle{abbrv}
\bibliography{Topochange}

\begin{thebibliography}{10}

\bibitem{AnDW}
A.~Anderson and B.~S. DeWitt.
\newblock {Does the Topology of Space Fluctuate?}
\newblock {\em Found. Phys.}, 16:91--105, 1986.

\bibitem{AGH}
L.~Andersson, G.~J. Galloway, and R.~Howard.
\newblock The cosmological time function.
\newblock {\em Classical Quantum Gravity}, 15(2):309--322, 1998.

\bibitem{BEE}
J.~K. Beem, P.~E. Ehrlich, and K.~L. Easley.
\newblock {\em Global {L}orentzian geometry}, volume 202 of {\em Monographs and
  Textbooks in Pure and Applied Mathematics}.
\newblock Marcel Dekker, Inc., New York, second edition, 1996.

\bibitem{BeSa1}
A.~N. Bernal and M.~S{\'a}nchez.
\newblock On smooth {C}auchy hypersurfaces and {G}eroch's splitting theorem.
\newblock {\em Comm. Math. Phys.}, 243(3):461--470, 2003.

\bibitem{BeSa2}
A.~N. Bernal and M.~S{\'a}nchez.
\newblock Smoothness of time functions and the metric splitting of globally
  hyperbolic spacetimes.
\newblock {\em Comm. Math. Phys.}, 257(1):43--50, 2005.

\bibitem{BDGSS}
A.~Borde, H.~F. Dowker, R.~S. Garcia, R.~D. Sorkin, and S.~Surya.
\newblock Causal continuity in degenerate spacetimes.
\newblock {\em Classical Quantum Gravity}, 16(11):3457--3481, 1999.

\bibitem{BDJS}
M.~Buck, F.~Dowker, I.~Jubb, and R.~Sorkin.
\newblock The {S}orkin{\textendash}{J}ohnston state in a patch of the trousers
  spacetime.
\newblock {\em Classical Quantum Gravity}, 34(5):055002, 2017.

\bibitem{BuGH}
A.~Burtscher and L.~Garc\'\i{}a-Heveling.
\newblock {Time functions on Lorentzian length spaces}.
\newblock \href{https://arxiv.org/abs/2108.02693}{arXiv:2108.02693[gr-qc]},
  2021.

\bibitem{Dow}
F.~Dowker.
\newblock {Topology change in quantum gravity}.
\newblock In {\em {Workshop on Conference on the Future of Theoretical Physics
  and Cosmology in Honor of Steven Hawking's 60th Birthday}}, pages 436--452, 6
  2002.

\bibitem{DoSu}
F.~Dowker and S.~Surya.
\newblock {Topology change and causal continuity}.
\newblock {\em Phys. Rev. D}, 58:124019, 1998.

\bibitem{DoGa}
H.~F. Dowker and R.~S. Garcia.
\newblock {A Handlebody calculus for topology change}.
\newblock {\em Classical Quantum Gravity}, 15:1859--1879, 1998.

\bibitem{DGS2}
H.~F. Dowker, R.~S. Garcia, and S.~Surya.
\newblock {$K$}-causality and degenerate spacetimes.
\newblock {\em Classical Quantum Gravity}, 17(21):4377--4396, 2000.

\bibitem{DGS1}
H.~F. Dowker, R.~S. Garcia, and S.~Surya.
\newblock Morse index and causal continuity. {A} criterion for topology change
  in quantum gravity.
\newblock {\em Classical Quantum Gravity}, 17(3):697--712, 2000.

\bibitem{GaHo}
S.~R. Garcia and R.~A. Horn.
\newblock {\em A Second Course in Linear Algebra}.
\newblock Cambridge Mathematical Textbooks. Cambridge University Press, 1
  edition, 2017.

\bibitem{Ger2}
R.~Geroch.
\newblock Domain of dependence.
\newblock {\em J. Mathematical Phys.}, 11:437--449, 1970.

\bibitem{Ger}
R.~P. Geroch.
\newblock {Topology in general relativity}.
\newblock {\em J. Math. Phys.}, 8:782--786, 1967.

\bibitem{HaSa}
S.~W. Hawking and R.~K. Sachs.
\newblock Causally continuous spacetimes.
\newblock {\em Comm. Math. Phys.}, 35:287--296, 1974.

\bibitem{Jan}
D.~W. {Janssen}.
\newblock {Quantum Fields on Semi-globally Hyperbolic Space-Times}.
\newblock {\em Comm. Math. Phys.}, 2022.
\newblock
  \href{https://doi.org/10.1007/s00220-022-04328-7}{https://doi.org/10.1007/s00220-022-04328-7}.

\bibitem{Joh}
S.~Johnston.
\newblock Feynman propagator for a free scalar field on a causal set.
\newblock {\em Phys. Rev. Lett.}, 103(18):180401, 4, 2009.

\bibitem{KoSe}
M.~Kontsevich and G.~Segal.
\newblock Wick rotation and the positivity of energy in quantum field theory.
\newblock {\em Q. J. Math.}, 72(1-2):673--699, 2021.

\bibitem{Kun}
W.~Kundt.
\newblock Non-existence of trouser-worlds.
\newblock {\em Comm. Math. Phys.}, 4(2):143--144, 1967.

\bibitem{KuSa}
M.~Kunzinger and C.~S\"{a}mann.
\newblock Lorentzian length spaces.
\newblock {\em Ann. Global Anal. Geom.}, 54(3):399--447, 2018.

\bibitem{LoSo}
J.~Louko and R.~D. Sorkin.
\newblock {Complex actions in two-dimensional topology change}.
\newblock {\em Classical Quantum Gravity}, 14:179--204, 1997.

\bibitem{MCT}
C.~Manogue, E.~Copeland, and T.~Dray.
\newblock {The trousers problem revisited}.
\newblock {\em Pramana}, 30:279--–292, 1988.

\bibitem{Mil}
J.~W. {Milnor}.
\newblock {\em {Morse theory. Based on lecture notes by M.\ Spivak and R.\
  Wells}}, volume~51 of {\em {Annals of Mathematics Studies}}.
\newblock Princeton University Press, 1963.

\bibitem{Min}
E.~Minguzzi.
\newblock {Limit curve theorems in Lorentzian geometry}.
\newblock {\em J. Math. Phys.}, 49:092501--092518, 2008.

\bibitem{San2}
M.~S{\'a}nchez.
\newblock {A class of cosmological models with spatially constant sign-changing
  curvature}.
\newblock {\em Port. Math.}, 2023.
\newblock DOI 10.4171/PM/2099.

\bibitem{Sor2}
R.~D. Sorkin.
\newblock Consequences of spacetime topology.
\newblock In {\em Proceedings of the Third Canadian Conference on General
  Relativity and Relativistic Astrophysics}, pages 137--163, 1989.

\bibitem{Sor}
R.~D. Sorkin.
\newblock {Forks in the road, on the way to quantum gravity}.
\newblock {\em Int. J. Theor. Phys.}, 36:2759--2781, 1997.

\bibitem{Sor3}
R.~D. Sorkin.
\newblock {Scalar Field Theory on a Causal Set in Histories Form}.
\newblock {\em J. Phys. Conf. Ser.}, 306:012017, 2011.

\bibitem{Whe}
J.~A. {Wheeler}.
\newblock {On the nature of quantum geometrodynamics}.
\newblock {\em {Ann. Phys.}}, 2:604--614, 1957.

\bibitem{Wit}
E.~Witten.
\newblock {A Note On Complex Spacetime Metrics}.
\newblock \href{https://arxiv.org/abs/2111.06514}{arXiv:2111.06514[hep-th]},
  2021.

\bibitem{Yod1}
P.~Yodzis.
\newblock Lorentz cobordism.
\newblock {\em Comm. Math. Phys.}, 26:39--52, 1972.

\bibitem{Yod2}
P.~Yodzis.
\newblock Lorentz cobordism. {II}.
\newblock {\em Gen. Relativity Gravitation}, 4(4):299--307, 1973.

\end{thebibliography}

\end{document}